\documentclass[11 pt]{article}
\usepackage[margin=1in]{geometry}

\usepackage{amsmath,amsthm,amsfonts,graphicx,float}
\usepackage{amssymb,upgreek,stmaryrd,framed,bbm}
\usepackage{wrapfig,mathrsfs,xfrac,mathtools,epsdice}
\usepackage{qtree,setspace,accents,verbatim}
\usepackage[x11names]{xcolor}
\usepackage[]{algorithm2e}
\usepackage[normalem]{ulem}
\usepackage[english]{babel}
\usepackage{framed}
\usepackage[shortlabels]{enumitem}

\usepackage{thmtools}
\usepackage{thm-restate}

\makeatletter\let\@@citation@@=\citation\renewcommand{\citation}[1]{\@@citation@@{#1}\@for\@tempa:=#1\do{\@ifundefined{cit@\@tempa}{\global\@namedef{cit@\@tempa}{}}{}}}\makeatother
\usepackage[linkbordercolor=red,citebordercolor=green,urlbordercolor=blue]{hyperref}
\usepackage[capitalize]{cleveref}
\makeatletter\def\@lbibitem[#1]#2#3\par{\@ifundefined{cit@#2}{}{\@skiphyperreftrue\H@item[\ifx\Hy@raisedlink\@empty\hyper@anchorstart{cite.#2\@extra@b@citeb}\@BIBLABEL{#1}\hyper@anchorend\else\Hy@raisedlink{\hyper@anchorstart{cite.#2\@extra@b@citeb}\hyper@anchorend}\@BIBLABEL{#1}\fi\hfill]\@skiphyperreffalse}\if@filesw\begingroup\let\protect\noexpand\immediate\write\@auxout{\string\bibcite{#2}{#1}}\endgroup\fi\ignorespaces\@ifundefined{cit@#2}{}{#3}} \def\@bibitem#1#2\par{\@ifundefined{cit@#1}{}{\@skiphyperreftrue\H@item\@skiphyperreffalse\Hy@raisedlink{\hyper@anchorstart{cite.#1\@extra@b@citeb}\relax\hyper@anchorend}}\if@filesw\begingroup\let\protect\noexpand\immediate\write\@auxout{\string\bibcite{#1}{\the\value{\@listctr}}}\endgroup\fi\ignorespaces\@ifundefined{cit@#1}{}{#2}}\makeatother

\newcommand{\Ex}{\mathbb{E}}
\newcommand{\F}{\mathbb{F}}

\newcommand{\PP}{\mathbb{P}}

\newcommand{\R}{\mathbb{R}}

\newcommand{\cB}{\mathcal{B}}

\newcommand{\cG}{\mathcal{G}}

\newcommand{\cM}{\mathcal{M}}

\newcommand{\bE}{\mathbb{E}}\newcommand{\bF}{\mathbb{F}}
\newcommand{\bG}{\mathbb{G}}

\newcommand{\bbone}{\boldsymbol{\mathbbm1}}

\newcommand{\eps}{\upvarepsilon}

\renewcommand{\rm}[1]{\mathrm{#1}}

\newcommand{\ceil}[1]{\left\lceil #1 \right\rceil}
\newcommand{\floor}[1]{\left\lfloor #1 \right\rfloor}

\newcommand{\smat}[1]{\begin{smallmatrix}#1\end{smallmatrix}}

\def\multichoose#1#2{\ensuremath{\left(\kern-.3em\left(\genfrac{}{}{0pt}{}{#1}{#2}\right)\kern-.3em\right)}}

\DeclareMathOperator{\var}{Var}

\DeclareMathOperator{\bin}{Bin}

\DeclareMathOperator{\rank}{rank}

\DeclareMathOperator{\RHS}{RHS}

\newtheorem{theorem}{Theorem}[section]
\newtheorem{lemma}[theorem]{Lemma}

\newtheorem{proposition}[theorem]{Proposition}
\newtheorem{corollary}{Corollary}[theorem]

\newtheorem{conjecture}[theorem]{Conjecture}

\newtheorem{innercustomclaim}{Claim} \newenvironment{numclaim}[1]
{\renewcommand\theinnercustomclaim{#1}\innercustomclaim}
  {\endinnercustomclaim}

\theoremstyle{definition}
\newtheorem*{definition}{Definition}
\newtheorem*{observation}{Observation}

\theoremstyle{remark}
\newtheorem*{remark}{Remark}

\newenvironment{subproof}[1][\proofname]{
\begin{proof}[#1]}{\end{proof}}



\newcommand{\gadget}{\mathbf{G}}
\newcommand{\gn}{G}
\DeclareMathOperator{\Adj}{Adj}
\newcommand{\apiv}{\mathbin{\raisebox{-0.3ex}{\(\tilde{\phantom{\times}}\)}\hspace{-0.8em}\times}}
\newcommand{\symdif}{\mathbin{\triangle}}
\newcommand{\disjcup}{\mathbin{\sqcup}}

\title{Almost all graphs are vertex-minor universal}

\author{Ruben Ascoli\thanks{School of Mathematics, Georgia Institute of Technology, Atlanta. Emails: \texttt{\{rascoli3, sfrederickson3, cmcfarland30, lpost3\}@gatech.edu}},
Bryce Frederickson\thanks{Department of Mathematics, Emory University, Atlanta, GA. Email: \texttt{bryce.frederickson@emory.edu}},
Sarah Frederickson$^*$\hspace{-.15cm},
Caleb McFarland$^*$\footnote{Supported in part by the National Science Foundation under Grant No. DMS- 2452111.}\phantom{ }\footnote{Supported in part by the Georgia Tech ARC-ACO Fellowship.},
Logan Post$^*$}

\begin{document}
\maketitle

\begin{abstract}
Answering a question of Claudet, we prove that the uniformly random graph $G\sim \bG(n, 1/2)$ is $\Omega(\sqrt n)$-vertex-minor universal with high probability. That is, for some constant $\alpha\approx 0.911$, any graph on any $\alpha\sqrt n$ specified vertices of $G$ can be obtained as a vertex-minor of $G$. This has direct implications for quantum communications networks: an $n$-vertex $k$-vertex-minor universal graph corresponds to an $n$-qubit $k$-stabilizer universal graph state, which has the property that one can induce any stabilizer state on any $k$ qubits using only local operations and classical communications. 

We further employ our methods in two other contexts. We obtain a bipartite pivot-minor version of our main result, and we use it to derive a universality statement for minors in random binary matroids.
We also introduce the \textit{vertex-minor Ramsey number} $R_{\rm{vm}}(k)$ to be the smallest value $n$ such that every $n$-vertex graph contains an independent set of size $k$ as a vertex-minor. Supported by our main result, we conjecture that $R_{\rm{vm}}(k)$ is polynomial in $k$. We prove $\Omega(k^2) \leq R_{\rm{vm}}(k) \leq 2^k - 1$.
\end{abstract}

\section{Background and main results}

Recent work in quantum information theory has emphasized the role of entangled multipartite states as reusable resources for large-scale quantum communication networks. In such networks, separate parties are restricted to Local Operations and Classical Communication (LOCC), with no ability to perform joint quantum operations once the resource state has been distributed. A central question in this setting is to understand which global entangled states allow pairs of parties to generate entangled EPR pairs.

Motivated by this perspective, Bravyi, Sharma, Szegedy, and de Wolf~\cite{Br+} introduced the notion of $k$-pairable quantum states. An $n$-party state $|\psi\rangle$ is $k$-pairable if, for any choice of $k$ disjoint pairs of parties, there exists a LOCC protocol that starts with $|\psi\rangle$ and ends in a state such that each of those pairs of parties share an EPR pair. This definition reflects a worst-case notion of network universality: a single fixed resource state must, under LOCC, generate EPR pairs for any choice of $k$ disjoint pairs of parties. Bravyi et al.~\cite{Br+} established upper and lower bounds on the achievable value of 
$k$ in terms of the number of parties and number of qubits per party.

Subsequently, Claudet, Mhalla, and Perdrix~\cite{CMP} and Cautr\`es, Claudet, Mhalla, Perdrix, Savin, and Thomass\'e~\cite{Cau+} elaborated on the ideas of Bravyi et al.~\cite{Br+}, focusing on the case where each party has just one qubit and studying a generalization of $k$-pairable states which Cautr\'es et al.~\cite{Cau+} call \emph{$k$-stabilizer universal} states. A state is $k$-stabilizer universal if any stabilizer state on any subset of $k$ qubits can be obtained through LOCC protocols. Stabilizer universality is a stronger notion than $k$-pairability: EPR pairs are stabilizer states, so any $2k$-stabilizer universal state is also $k$-pairable. For more background on stabilizer states, see for example \cite{Hein2006}.

The contributions of \cite{CMP} and \cite{Cau+}, as well as those of the current paper, rely on the connection between quantum states and graph theory. Central to this connection is the notion of vertex-minors, defined as follows.
\begin{definition}
Given a graph $G$ and a vertex $v\in V(G)$, a \emph{local complementation} at $v$ is an operation which replaces $G[N(v)]$, the subgraph of $G$ induced on the neighborhood of $v$, by its complement. The resulting graph is denoted $G*v$ and has the same vertex set as $G$. 
A graph $H$ with $V(H)\subseteq V(G)$ is a \emph{vertex-minor} of $G$ if it can be obtained from $G$ through a sequence of local complementations and vertex deletions.
\end{definition}
Vertex-minors capture an important sequence of operations on quantum states: if $H$ can be obtained as a vertex-minor of $G$, then the quantum graph state corresponding to $H$ can be obtained from the graph state corresponding to $G$ through just local Clifford operations, local Pauli measurements, and classical communications. The converse also holds if $H$ has no isolated vertices \cite{DW}.

Besides serving as a way to characterize certain properties of quantum states, vertex-minors have been the star of numerous fascinating problems in graph theory since the 1980s. They often serve as a dense analogue of the usual minor relation, and so similarly have many deep applications to width parameters, algorithms, and logic \cite{Oum05, geelen2023grid, courcelle2007vertex, buffiere2024shallow}. For more background on vertex-minors and their applications, see the recent survey on vertex-minors by Kim and Oum~\cite{KO}.

The following definition, introduced by Claudet et al.~\cite{CMP} and further studied by Cautr\`es et al.~\cite{Cau+}, provides the bridge between the purely graph-theoretic concept of vertex-minors and the notion of $k$-stabilizer universal states.

\begin{definition} A graph $G$ is \emph{$k$-vertex-minor universal} on $V(G)$ if for every vertex subset $U\subseteq V(G)$ with $|U|=k$, every graph $H$ on $U$ is a vertex-minor of $G$. 
\end{definition}

We emphasize that we require each $H$ to be obtained exactly on each specified vertex set $U$. In particular, there are $\binom{|V(G)|}{k}2^{\binom{k}{2}}$ total choices of pairs $U$ and $H$. If $G$ is a $k$-vertex-minor universal graph then the quantum graph state corresponding to $G$ is $k$-stabilizer universal. This enables explicit constructions of universal resource states for quantum communication protocols \cite{CMP, Cau+}.

Our main result, resolving a question of Claudet~\cite{Cla}, is that for some constant $C$, \emph{almost all} graphs on $Ck^2$ vertices are $k$-vertex-minor universal. That is, a uniformly random graph on $Ck^2$ vertices is $k$-vertex-minor universal with probability tending exponentially fast to $1$ as $k$ increases. The quadratic dependence is best possible up to the constant factor $C$. We use the notation $\bG(n, p)$ for the Erd\H{o}s-Renyi distribution on graphs on $n$ vertices where each edge is included independently with probability $p$; in particular, $\bG(n,1/2)$ is the uniform distribution on graphs on $n$ vertices. 

\begin{restatable}{theorem}{vmunivthm}\label{thm:vm_universal} For all $c>0$, if $G\sim \bG(n,1/2)$ with $n\geq (1+c)\frac{1}{2\log_2(4/3)} k^2$, then $G$ is $k$-vertex-minor universal with probability at least $1-2^{-(1+o(1))ck^2/2}$. 
\end{restatable}

On the practical side, our explicit bound on the probability of success (see \eqref{eq:vm-final-bound}) converges exponentially quickly beyond the threshold value. Thus we have a probabilistic algorithm to produce a $k$-vertex-minor universal graph on $O(k^2)$ vertices. Explicitly, letting $C=\frac{1}{2\log_2(4/3)}\approx 1.205$, we have that for any $\eps>0$ and $k\geq 3$, a uniformly random graph on $n=\ceil{Ck^2+16k\log k+4\log(1/\eps)}$ vertices is $k$-vertex-minor universal with probability at least $1-\eps$. This follows from keeping track of the error terms in the union bound. The proof of \cref{thm:vm_universal} is non-constructive, and we do not obtain an algorithm for finding a desired vertex-minor in $\bG(n,1/2)$. We note that this problem is NP-hard in general \cite{DAHLBERG2022106222,DHW}.

Although the study of $k$-vertex-minor universal graphs has only taken off in recent years, the relationship between quantum graph states and vertex-minors dates back to 2004 with the pioneering work of 
Van den Nest, Dehaene, and De Moor \cite{VdN}. In the years since,
many authors have studied various models of random graph states and their algorithmic and quantum entanglement properties; see e.g. \cite{Wu+, DHW, KH, GHH} and the references therein for non-exhaustive coverage of the subject.

Specifically regarding the notion of $k$-vertex-minor universal graphs, Claudet et al.~\cite{CMP} showed that there exist $k$-vertex-minor universal graphs with $O(k^4\log k)$ vertices, while any $k$-vertex-minor universal graph must have at least $(k-2)^2/(2\log_2 3)$ vertices. Cautr\`es et al.~\cite{Cau+} improved the existence result to match the lower bound up to a constant factor, establishing the existence of $k$-vertex-minor universal graphs on $(2+o(1))k^2$ vertices. Their construction is based on a random unbalanced bipartite graph $G$ where one side has $\Theta(k^2)$ vertices and the other has $\Theta(k\log k)$ vertices, with edges between the two sides included independently with probability $1/2$. This yields a graph with $\Theta(k^3\log k)$ edges with high probability, which is less than a constant proportion of possible edges. They also give a polynomial time algorithm to find a $k$-vertex graph $H$ as a vertex-minor in $G$ with high probability.

Our \cref{thm:vm_universal} improves the constant factor $2$ from Cautr\`es et al.~\cite{Cau+} to $\frac{1}{2\log_2(4/3)} \approx 1.205$, progressing closer to the lower bound~\cite{CMP} of $1/(2\log_2 3) \approx 0.315$.
More conceptually important, however, is that our result states that almost all graphs -- including dense and non-bipartite graphs -- also have the vertex-minor universality property. This is perhaps surprising in the context of a paper of Gross, Flammia, and Eisert~\cite{GFE}, which shows that all but an exponentially small fraction of quantum states are too entangled to be computationally useful. Our theorem, in contrast, shows that all but an exponentially small fraction of $O(k^2)$-qubit \emph{graph} states are in fact $k$-stabilizer universal.  

Although the vertex-minor problem addressed by \cref{thm:vm_universal} is most relevant to quantum information theory, our methods extend to a stronger notion of universality. The following definition plays a key role.

\begin{definition}
Given a graph $G$ and an edge $uv\in E(G)$, a \textit{pivot} at the edge $uv$ is the operation which replaces $G$ by $G \times uv \coloneq G * u * v * u$. It is easy to verify that $G * u * v * u = G * v * u * v$, so the operation is well-defined. 
A graph $H$ with $V(H)\subseteq V(G)$ is a \emph{pivot-minor} of $G$ if it can be obtained from $G$ through a sequence of pivots and vertex deletions. A graph $G$ is \emph{$k$-pivot-minor universal} on $V(G)$ if for every vertex subset $U\subseteq V(G)$ with $|U|=k$, every graph $H$ on $U$ is a pivot-minor of $G$. 
\end{definition}

Our techniques naturally extend to pivot-minors. Pivot-minors in bipartite graphs are essentially equivalent to minors in binary matroids, and thus we also obtain the following two results. See \Cref{section:Universality for pivot-minors and binary matroids} for more background, precise definitions for matroid minor universality, and proofs.

\begin{theorem}\label{thm:ER_pm_universal}For all $c>0$, if $G\sim \bG(n,1/2)$ with $n\geq (1+c)\frac{2}{\log_2(16/13)}k^2\approx (1+c)6.68k^2$, then $G$ is $k$-pivot-minor universal with probability at least $1-2^{-(1+o(1))ck^2}$.
\end{theorem}

\begin{restatable}{theorem}{matroidUniversal}\label{thm:matroid_universal}
Let $C=\frac{1}{2\log_2(8/7)}+\frac 14\approx 2.85$. For all $c>0$, if $M$ is uniformly random binary matroid on ground set $[n]$ and $n\geq (1+c)2Ck^2$, then $M$ is $k$-minor universal with probability at least $1-2^{-(1+o(1))ck^2/4}$.
\end{restatable}

For $G$ to be $k$-vertex-minor universal requires that we can obtain every graph on every $k$-subset of vertices of $G$ as a vertex-minor. We can also ask a related question: how big does $n$ need to be to ensure that \emph{every} graph on $n$ vertices contains a \emph{particular} graph $H$ as a vertex-minor on some $k$ vertices? The answer is only finite if the graph $H$ has no edges, or is an \emph{independent set}; indeed, otherwise, a graph of any size $n$ with no edges does not contain $H$ as a vertex-minor. 
\begin{definition}
Let $R_{\rm{vm}}(k)$ be the minimum $n$ such that any $n$-vertex graph contains the independent set of size $k$ as a vertex-minor.
\end{definition}
Recall the classical Ramsey number $R(s,t)$ which is the smallest $n$ such that any $n$-vertex graph contains either an independent set of size $s$ or a clique of size $t$. As observed in \cite{Cam+}, we have $R_{\rm{vm}}(k) \leq R(k,k+1)$, since if a graph contains a clique of size $k+1$, locally complementing at one vertex forms an independent set on the remaining $k$ vertices. The current best upper bound on $R(k,k+1)$ is roughly $3.8^k$, due to Gupta, Ndiaye, Norin, and Wei~\cite{GNNW} following the breakthrough work of Campos, Griffiths, Morris, and Sahasrabudhe~\cite{CGMS}. Here we provide a nontrivial, but still exponential upper bound for $R_{\rm{vm}}(k)$, as well as a quadratic lower bound obtained from the random graph.

\begin{restatable}{theorem}{ramseybounds}\label{thm:ramsey_bounds} For all $k\geq 0$, 
\[
(1+o(1))\frac{1}{2\log_2 3} k^2\leq R_{\rm{vm}}(k)\leq 2^k-1.
\]
\end{restatable}

The lower bound comes from showing that with high probability, the uniformly random graph $G\sim \bG(n,1/2)$ does not contain an independent set of size $k$ as a vertex-minor if $k \geq (1+\eps)\sqrt{2\log_2(3)n}$, for any $\eps > 0$. Due to \cref{thm:vm_universal}, this is best possible up to the constant factor, and we believe the lower bound is closer to the true answer.

\begin{restatable}{conjecture}{ramseyconj}\label{conj:ramsey_conj} $R_{\rm{vm}}(k) \leq \mathsf{poly}(k)$.
\end{restatable}

The rest of this paper is organized as follows. In \cref{sec:prelims}, we provide the technical graph-theoretic and probabilistic tools we need for the proof of \cref{thm:vm_universal}, state some preliminary lemmas, and give a detailed overview of the proof. In \cref{section:proofs}, we carry out our proof strategy for \cref{thm:vm_universal}. Then, in \cref{section:Universality for pivot-minors and binary matroids}, we shift our focus to pivot-minors and binary matroids, providing background and proofs of \cref{thm:ER_pm_universal} and \cref{thm:matroid_universal}. Finally, in \cref{section:VM Ramsey conclusion}, we prove \cref{thm:ramsey_bounds} about the vertex-minor Ramsey number and conclude with some remaining open questions.

\section{Preliminaries and proof overview}\label{sec:prelims}
We denote the symmetric difference of sets by $A \symdif B = (A \setminus B) \cup (B \setminus A)$. For a graph $G$, we use $V(G)$ and $E(G)$ to denote its vertex set and edge set respectively. For $U \subseteq V(G)$, we use $G[U]$ to denote the graph induced on $U$, and $G - U$ to denote the graph obtained from deleting $U$. The \emph{adjacency matrix} of a graph $G$ (over $\mathbb{F}_2$) is the symmetric matrix $A \in \mathbb{F}_2^{V(G) \times V(G)}$ such that $A_{u,v} = 1$ if $uv \in E(G)$ and $A_{u,v} = 0$ otherwise. Note that as we assume graphs to be simple, $A$ has 0 along the diagonal. For disjoint sets $X,Y$, we define the set $[X,Y]\coloneqq \{\{x,y\}\mid x\in X, y\in Y\}$ and let $G[X,Y]$ be the induced bipartite subgraph with edge set $E(G)\cap [X,Y]$.

\subsection{Vertex-minor preliminaries}
Throughout the paper it will be convenient to have notation which represents both the local complementation and pivot operations. To that end we define the following. Let $x \in \{\varnothing, \{v\}, \{u,v\}\}$. Then
$$G \circ x = \begin{cases}
    G & x = \varnothing\\
    G * v & x = \{v\}\\
    G \times uv & x = \{u,v\}.
\end{cases}$$
Implicitly we assume that if $x = \{u,v\}$, then $uv \in E(G)$, otherwise the operation is not defined. If $(x_i)_{i=1}^\ell$ is a sequence of such $x$, then we often use $G \circ (x_i)_{i=1}^\ell$ as shorthand for $G \circ x_1 \circ x_2 \circ \dots \circ x_n$. We similarly write $G * (v_i)_{i=1}^\ell$ as shorthand for $G*v_1 * \dots * v_\ell$.

We now discuss some well-known facts about vertex-minors and pivot-minors. We say two graphs are \textit{locally equivalent} if one can be obtained from the other by a sequence of local complementations. Similarly we say that two graphs are \textit{pivot equivalent} if one can be obtained from the other by a sequence of pivots. It is easy to see that if $H$ is a vertex- (pivot-) minor of $G$, then $H\subseteq G'$ for some graph $G'$ which is locally (pivot) equivalent to $G$. The following fact is well-known.

\begin{lemma}[{\cite[Proposition 2.5]{Oum05}}]\label{lem:cancel pivots}
    If $vu \in E(G)$, then $uw \in E(G \times vu)$ if and only if $vw \in E(G)$. Furthermore, if both edges exist, we have
    $$G \times vu \times uw = G \times vw.$$
\end{lemma}

The following result, proven independently by Bouchet~\cite{Bou88} and Fon-Der-Flaass~\cite{FDF88}, will be crucial in reordering the operations when finding a vertex-minor.

\begin{theorem}[\cite{Bou88, FDF88}]\label{thm:structure of minors}
    Suppose $H$ is a vertex-minor of $G$ and $v \in V(G) \setminus V(H)$. Then $H$ is a vertex-minor of
    \begin{enumerate}
        \item $G - v$,
        \item $G * v - v$, or
        \item $G \times vu - v$ for some $u \in N(v)$.
    \end{enumerate}
\end{theorem}

We note that \cref{lem:cancel pivots} implies that the choice of $u$ does not matter up to pivot-equivalence. That is, if $H$ is a vertex-minor of $G \times vu - v$ for some $u \in N(v)$, then $H$ is a vertex-minor of $G \times vu' - v$ for every $u' \in N(v)$. This helps us prove our key structural lemma below.
We call this our ``reordering lemma," because it allows us to reorganize a sequence of local complementations in a convenient way.

\begin{restatable}{lemma}{minorreorderinglemma}\label{lem:minor-reordering} Fix a set $\hat V$. Let $v_1,\ldots,v_\ell\in \hat V$ be a (possibly repeating) sequence and let $\hat G$ be a graph with $V(\hat G)=\hat V$. Then there is a sequence $x_1,\ldots, x_r$ of singletons and pairs from $\hat V$ such that for any graph $G$ with $\hat V\subseteq V(G)$ and $G[\hat V]=\hat G$, we have
\begin{enumerate}[(I)]
\item[\emph{(I)}] $G*v_1*v_2\cdots *v_\ell-\hat V=G\circ x_1\cdots \circ x_r-\hat V$.
\item[\emph{(II)}] The sets $x_1,\ldots, x_r$ are pairwise disjoint.
\item[\emph{(III)}] If $v\in \hat V$ appears exactly once in $(v_i)_{i=1}^\ell$, then $v\in x_j$ for some $j\in [r]$.
\end{enumerate}
Moreover, any choice of $(x_i)_i$ satisfying (I) and (II) for all $G$ also satisfies (III).
\end{restatable}

\noindent 
We emphasize that the sequence $(x_i)_i$ depends \textit{only} on the subgraph $\hat G=G[\hat V]$, and is independent of the edges between $\hat V$ and $V(G)\setminus \hat V$.
Without this careful tracking of the dependency, this lemma would be implied by known results, (see \cite[Proposition 6]{CPLocal} or \cite[Proposition 7]{CPDeciding}). In \cref{section:proofofreorderinglemma}, we give a short proof by induction using \cref{thm:structure of minors}.

Because of the dependence only on $\hat{G}$, if $G = \bG(n,1/2)$, then by revealing only the edges of $\hat G$ we can rewrite any arbitrary sequence of complementations $*(v_i)_i$ in the desired form $\circ (x_j)_j$, and then study the statistics of the resulting graph.

\subsection{Probability preliminaries}

We use $\bG(n, p)$ to denote the Erd\H{o}s-R\'{e}nyi random graph distribution on vertex set $[n]$, where each edge is present independently with probability $p$. We let $\cG_n$ denote the set of all labeled graphs on vertex set $[n]$. The following simple proposition is crucial for our proof.

\begin{proposition}
\label{prop:complementation preserves random}If $G\sim \bG(n,1/2)$ and $v\in[n]$, then $G*v\sim \bG(n,1/2)$.
\end{proposition}
\begin{proof} Since the map $G\mapsto G*v$ is an involution on $\cG_n$, it preserves the uniform distribution. Indeed, for any graph $G' \in \cG_n$, we have $\mathbb P[G*v=G'] = \mathbb P[G=G'*v] = \frac{1}{|\cG_n|}$.
\end{proof}

The basic probabilistic tool we rely on is Chebyshev's inequality, which we use in the following form.
\begin{lemma}[Chebyshev's inequality]\label{lem:chebyshev} For a random variable $X$ and $\alpha>0$, we have 
\[\PP[|X-\Ex[X]|\geq \alpha]\leq \frac{\var[X]}{\alpha^2}.\]
In particular, setting $\alpha=\Ex[X]$ gives $\PP[X=0]\leq \frac{\var[X]}{\Ex[X]^2}$.
\end{lemma}

\subsection{Proof overview}\label{sec:overview}
Before the complete proof in \cref{section:proofs}, we outline the argument for \cref{thm:vm_universal}.  Let $G\sim \bG(n+k,1/2)$, where $k=\Theta(\sqrt n)$. Fix a vertex set $U\in {V(G)\choose k}$ and a graph $H$ on $U$. Denote the vertices of $G$ by $V=\{v_1,\ldots, v_n\}\disjcup U$. For each subset $J\subseteq [n]$, where $J=\{i_1<i_2<\ldots<i_{|J|}\}$ we define 
\begin{align}\label{eq:GJ defn}
G_J\coloneqq G*v_{i_1}*v_{i_2}*\cdots *v_{i_{|J|}},
\end{align}
so we can track $2^n$ distinct procedures for complementing on $G$. For example, $G_\emptyset=G$ and $G_{[n]}=G*v_1*\cdots *v_n$. Now, define the variable $X_J=\bbone_{G_J[U]=H}$ as the indicator for the event $\{G_J[U]=H\}$. Let $X=\sum_{J\subseteq [n]}X_J$. By \cref{prop:complementation preserves random}, for each $J$ we have $G_J\sim \bG(n,1/2)$, so each $X_J$ is 1 with the same probability. Thus $\Ex[X]=2^n2^{-{k\choose 2}}\to \infty$, and we want to show that $X>0$ (and thus $H$ is a vertex-minor of $G$) with \textit{very} high probability using Chebyshev's inequality. Eventually, we apply the union bound over all choices of $U$ and $H$ to show that $G$ is $k$-vertex-minor universal.

The central challenge for this argument is to show that the various $G_J,G_{J'}$ are not too correlated, and hence to control $\var[X]$. For example, comparing $G_\emptyset=G$ and $G_{[n]}$, we find that each time we complement at a vertex $v_i$, a random subset of the edges in $U$ are flipped. Hence, we expect that after sufficiently many flips, the distributions of $G[U]$ and $G_{[n]}[U]$ will be almost entirely uncorrelated. 

To make this precise, we define two random walks on the space of graphs with vertex set $U$ (see \cref{section:randomwalks}). The first random walk chooses a uniformly random subset of $U$ and complements the edges with both ends in that subset. This corresponds to complementing at a vertex outside of $U$ that has a random neighborhood. The second random walk chooses two uniformly random subsets $S,T\subseteq U$ and toggles adjacency between pairs of vertices which lie in two different sets among $S \setminus T, T \setminus S$, and $S \cap T$. This corresponds to pivoting on an edge outside $U$ where each end has a random neighborhood in $U$.

We view these random walks as linear transformations over the space of all graphs, and we compute their eigenvectors and eigenvalues exactly (see \cref{lem:eigenvalues_of_walk_explicit}). We use this spectral calculation to conclude that both random walks quickly approach the uniform distribution (\cref{lem:random_walk_mixing}).

We then use these random walks to study $G_J[U] \symdif G_{J'}[U]$ when $J$ and $J'$ differ. Specifically, we argue that if $|J \symdif J'|$ is large, then $G_J[U] \symdif G_{J'}[U]$ is close to uniformly random, and thus the dependence between $X_J$ and $X_{J'}$ is very small (see \cref{lem:JJ' are mixed}). Suppose $J = \{i_1 < i_2 < \dots < i_{|J|}\}$ and $J' = \{j_1 < j_2 < \dots < j_{|J'|}\}$. Because local complementation is an involution, we have
\[G_{J'} = G_J * v_{i_{|J|}} * \dots * v_{i_1} * v_{j_1} * \dots * v_{j_{|J'|}},\]
meaning that $G_{J'}[U]$ is formed from $G_J[U]$ by complementing on vertices outside of $U$. This sounds like our random walk, but as stated, the same vertex could be complemented on multiple times, and thus the steps in the corresponding random walk may be correlated. To combat this, we apply our reordering \cref{lem:minor-reordering} to obtain a sequence $(x_i)_{i=1}^\ell$ where
\[G_{J'}[U] = G_J \circ  (x_i)_{i = 1}^\ell[U]\]
such that the $x_i$ are pairwise disjoint. Then we may consider $G_J[U] \triangle (G_J \circ  (x_i)_{i=1}^t)[U]$ as a random walk starting at $t=0$ with the empty graph and ending at $t = \ell$ with $G_J[U] \symdif G_{J'}[U]$. Property (III) of \cref{lem:minor-reordering} implies that $|\bigcup_{i=1}^\ell x_i| \geq |J \symdif J'|$, so if $|J \symdif J'|$ is sufficiently large, this random walk mixed for many steps, implying that $G_{J}[U] \symdif G_{J'}[U]$ is approximately uniformly random. Since typical sets $J,J'$ differ in $n/2$ many indices, this gives the desired correlation bounds, implying that the variance of $X$ is small and hence the probability that $H$ is a vertex-minor of $G$ is large (\cref{lem:vm_high_probability}).

We note some subtlety in the sets of edges which must stay independent in this approach. In order to apply \cref{lem:minor-reordering}, we first reveal the edges of $G \setminus U$. Then we reveal the edges between $x_i$ and $U$ one at a time so that at each step, each vertex in $x_i$ has a uniformly random neighborhood in $U$. We use \cref{lem:complementing_independence} to argue that conditional on all previous steps, the unexposed edges remain uniformly random throughout the process. This guarantees that each step of the random walk is independent of the previous steps, and so each step brings us closer to the uniform distribution.

\section{Proof of vertex-minor universality of the random graph}\label{section:proofs}

\subsection{Random walks on the space of graphs}\label{section:randomwalks}

We first define and consider two random walks $\rm{Com}$ and $\rm{Piv}$ on the space of graphs with vertex set $U=[k]$. We begin with a graph $\Gamma$ drawn from some starting distribution $\mu_0$ on all graphs with vertex set $U$. Then at each time $t$, the $\rm{Com}$ walk simulates locally complementing (and deleting) at vertex $v$ outside of $U$ with a uniformly random neighborhood. Similarly, the $\rm{Piv}$ walk simulates pivoting on an edge outside $U$ with uniformly random neighborhoods. Formally, we define the random walks as follows.

\begin{definition}
Suppose $\mu$ is a probability distribution on $\cG_k$, the set of labeled graphs with vertex set $[k]$. Identify $\mu$ with a vector in $[0,1]^{\cG_k}$, and we write $\mu(H)$ to be the probability of the graph $H \in \cG_k$. First, let $\rm{Com}(\mu)$ be the distribution of the graph $\Gamma\symdif K_S$, where $\Gamma\sim \mu$, and $S\subseteq [k]$ is a uniformly random subset independent of $\Gamma$. Next, for $S,T\subseteq [k]$ we let $K^\triangle_{S,T}$ be the complete tripartite graph with tripartition $\{S\setminus T,T\setminus S,S\cap T\}$. Then let $\rm{Piv}(\mu)$ be the distribution of $\Gamma\symdif K_{S,T}^\triangle$ where $\Gamma\sim \mu$ and $S,T\subseteq [k]$ are uniformly random and independent of $\Gamma$. Finally, let $\mu_{\rm{unif}}$ be the uniform distribution on $\cG_k$.
\end{definition}

It is easy to see that the uniform distribution $\mu_{\rm{unif}}$ is stationary under both the $\rm{Com}$ and $\rm{Piv}$ operators, and we will be interested in how fast the distribution $\mu_0$ approaches $\mu_{\rm{unif}}$ under these operators. It is helpful to think of these as steps in a random walk on a Cayley graph of the additive group $\bF_2^{E(K_k)}$. For example, if we construct the Cayley graph with multiset of generators $\{\bbone_{E(K_S)}:S\subseteq [k]\}$, then one step in this walk gives the convolution $\mu\mapsto \rm{Com}(\mu)$. A nice survey of random walks on groups is given by Diaconis \cite{Dia03}. As symmetric random walks, both $\rm{Com},\rm{Piv}:\R^{\cG_n}\to \R^{\cG_n}$ are symmetric linear transformations over $\R$. Because they are walks on a Cayley graph, we can explicitly compute their eigenvectors.

\begin{lemma}\label{lem:eigenvalues_of_walk_explicit}
    The eigenvectors of $\rm{Com}$ and $\rm{Piv}$ are given by $\{\chi_G \mid G \in \cG_k\}$ where $\chi_G \in \mathbb{R}^{\cG_k}$ is defined as $\chi_G(H) =(-1)^{|E(G) \cap E(H)|}$. For each $G \in \cG_k$, let $\lambda^{\rm{Com}}_G, \lambda^{\rm{Piv}}_G$ denote the eigenvalues of $\chi_G$ under $\rm{Com}, \rm{Piv}$ respectively. Then we have
    \begin{align*}
        \lambda^{\rm{Com}}_G &= \Ex_{S}\left[(-1)^{|E(G) \cap E(K_S)|}\right]\\
        \lambda^{\rm{Piv}}_G &= \Ex_{S,T}\left[(-1)^{|E(G) \cap E(K^\triangle_{S,T})|}\right]
    \end{align*}
    where the expectation is over $S,T \subseteq [k]$ drawn uniformly at random.
\end{lemma}
\begin{proof}
    The statement for $\rm{Com}$ follows from a simple calculation:
    \begin{align*}
        \rm{Com}(\chi_G)(H)&=2^{-k}\sum_{S\subseteq [k]}\chi_G(H\symdif K_S)\\
        &=2^{-k}\sum_{S\subseteq [k]}(-1)^{|E(G)\cap E(H)|}(-1)^{|E(G)\cap E(K_S)|}\notag\\
        &=\chi_G(H)\cdot 2^{-k}\sum_{S\subseteq [k]}(-1)^{|E(G)\cap E(K_S)|}\\
        &= \chi_G(H) \cdot \Ex_{S}\left[(-1)^{|E(G) \cap E(K_S)|}\right].
    \end{align*}
    The proof for $\rm{Piv}$ is similar.
\end{proof}

We can now bound the eigenvalues as follows.

\begin{lemma}\label{lem:eigenvalue_rank}
Let $G \in \cG_k$ have adjacency matrix $A \in \mathbb{F}_2^{k \times k}$. Then
\[
(\lambda^{\rm{Com}}_G)^2 \leq  
\lambda^{\rm{Piv}}_G = 2^{-\rm{rank}_{\mathbb{F}_2}(A)}.
\]
\end{lemma}

\begin{proof}
    We first compute $\lambda^{\rm{Piv}}_G$. Recall that for $S,T\subseteq [k]$, if $x_S$ and $x_T^{}$ are the indicator vectors of $S$ and $T$, then $x_S^\top Ax_T^{}=|\{(u,v)\mid u\in S,v\in T,uv\in E(G)\}|$ over $\mathbb{R}$. For example, $x_S^\top Ax_S^{}= 2|E(G) \cap E(K_S)|$. We observe that $|E(G) \cap E(K_{S,T}^\triangle)|\equiv x_S^\top Ax_T\pmod 2$. This is because every edge from $S$ to $T$ is counted, with edges in $S\cap T$ counted twice. Using this and \cref{lem:eigenvalues_of_walk_explicit},
    \begin{equation*}
        \lambda_G^{\rm{Piv}}=\Ex\left[(-1)^{e(G\cap K_{S,T}^\triangle)}\right]=2^{-2k}\sum_{S,T\subseteq [k]}(-1)^{x_T^\top Ax_S^{}}=2^{-2k}\sum_{u\in \bF_2^k}\bigg(\sum_{v\in\bF_2^k}(-1)^{v^\top Au}\bigg).
    \end{equation*}
    Consider the inner sum. If $u\in \rm{Null}_{\bF_2}(A)$, then (over $\bF_2$) we have $Au=0$, so the inner sum is $2^k$. Otherwise, $Au$ is a nonzero vector, so exactly half of vectors $v\in \bF_2^k$ have $v^\top Au=0$, and hence the inner sum is $0$. Thus,
    \begin{equation*}
        \lambda^\rm{Piv}_G= 2^{-2k}\Big(2^k|\rm{Null}(A)|\Big)=2^{-k}\cdot 2^{k-\rm{rank}(A)}=2^{-\rm{rank}(A)}.
    \end{equation*}
    Next, we use this to bound $\lambda^{\rm{Com}}_G$. For subsets $S,T\subseteq [k]$, let $R=S\symdif T$. Observe the following two identities:
    \begin{align*}
        |E(G) \cap E(K_S)|+ |E(G) \cap E(K_T)| &\equiv |E(G) \cap E(K_{S\symdif T})| + |E(G) \cap E(K_{S,T}^\triangle)| \pmod 2\\
        &\equiv |E(G) \cap E(K_R)| + |E(G) \cap E(K_{S,R}^\triangle)| \pmod 2.
    \end{align*}
    The first can be checked combinatorially, or using linear algebra (over $\mathbb R$):
    $$\frac{x_S^\top Ax_S}{2} + \frac{x_T^\top Ax_T}{2} = \frac{(x_S+x_T)^\top A(x_S+x_T)}{2} - x_T^\top A x_S.$$
    Evaluating each side mod 2 yields the first identity. For the second, note that by definition $K_{S,T}^\triangle=K_{S,R}^\triangle$, or using algebra $x_T^\top A x_S^{}\equiv (x_T+x_S)^\top Ax_S\pmod 2$. As a result, we have the following change of variables (where again we use \cref{lem:eigenvalues_of_walk_explicit}):
    \begin{align*}
        (\lambda_G^{\rm{Com}})^2=\Ex\big[(-1)^{|E(G) \cap E(K_S)|}\big]^2&=2^{-2k}\sum_{S\subseteq [k]}(-1)^{|E(G) \cap E(K_S)|}\sum_{T\subseteq [k]}(-1)^{|E(G) \cap E(K_T)|}\\
        &=2^{-2k}\sum_{S,T\subseteq [k]}(-1)^{|E(G) \cap E(K_{S\symdif T})|}(-1)^{|E(G) \cap E(K_{S,T}^\triangle)|}\\
        &=2^{-2k}\sum_{R\subseteq [k]}(-1)^{|E(G) \cap E(K_R)|}\sum_{S\subseteq [k]}(-1)^{e(G\cap K_{S,R}^\triangle)}\\
        &\leq 2^{-2k}\sum_{R\subseteq [k]}\bigg|\sum_{S\subseteq [k]}(-1)^{x_S^\top Ax_R}\bigg|
        =\lambda^{\rm{Piv}}_G.
    \end{align*}
In the last line, we used the fact again that the inner sum is either $2^k$ or $0$. Thus, we conclude that $(\lambda^{\rm{Com}}_G)^2\leq \lambda^{\rm{Piv}}_G =2^{-\rank(A)}$, as desired.
\end{proof}

\begin{remark}
    We note that $A$ is an alternating matrix over $\mathbb{F}_2$, and so the rank of $A$ is even. Thus in particular, if $G = I_k$, $\lambda^{\rm{Com}}_G = \lambda^{\rm{Piv}}_G = 1$, and for every other $G \in \cG_k$, $\lambda^{\rm{Com}}_G \leq 1/2$ and $\lambda^{\rm{Piv}}_G \leq 1/4$. 
\end{remark}

We note that $\mu_{\rm{unif}}$ is a multiple of $\chi_{I_k}$, and so the crude bound of $1/2$ is enough to get exponential convergence to the uniform distribution. To get even sharper bounds, we count the number of graphs whose adjacency matrix has a given rank. This is explicitly enumerated by MacWilliams~\cite{MacWilliams1969Orthogonal}, but we only require the following bound.

\begin{theorem}[{\cite[Theorem 3]{MacWilliams1969Orthogonal}}]\label{thm:graphrank_counting}
    Let $\cG_{k,r}$ denote the set of labeled graphs whose adjacency matrix (a symmetric matrix in $\mathbb{F}_2^{k \times k}$ with 0 diagonal) has rank $r$. Then if $r \geq 1$ is even,
    $$|\cG_{k,r}| \leq 2^{rk - 2}$$
    and if $r$ is odd, $|\cG_{k,r}| = 0$.
\end{theorem}

We are now ready to show that the two random walks (or any combination of the two) converge exponentially fast to the uniform distribution.

\begin{lemma}\label{lem:random_walk_mixing}
    Let $\mu_0 = \delta_{I_k}$ be the point mass distribution on the independent graph $I_k \in \cG_k$. Then for $t = 1, \dots, \ell$, define either $\mu_t = \rm{Com}(\mu_{t-1})$ or $\mu_t = \rm{Piv}(\mu_{t-1})$. Suppose that the $\rm{Com}$ transformation is applied $m_1$ times and the $\rm{Piv}$ transformation is applied $m_2$ times, and let $m=m_1+2m_2$. Then if $m>2k$, we have
    $$|\mu_\ell(H) - \mu_{\rm{unif}}(H)| \leq 2^{2k - \binom{k}{2} - m}$$
    for every graph $H \in \cG_k$.
\end{lemma}
\begin{proof}
    Note that $\{\chi_G \mid G \in \cG_k\}$ is an orthogonal basis of $\mathbb{R}^{\cG_k}$. Because $\mu_0=\delta_{I_k}$ has the same inner product with each eigenvector, we have exactly $\mu_0=2^{-{k\choose 2}}\sum_{G\in \cG_k}\chi_G$. By induction on $t$, we compute
    $$\mu_\ell = 2^{-\binom{k}{2}} \sum_{G \in \cG_k} (\lambda_G^{\rm{Com}})^{m_1} (\lambda_G^{\rm{Piv}})^{m_2} \chi_G.$$
    We note that $2^{-{k\choose 2}}\chi_{I_k}=\mu_{\rm{unif}}$ and $\lambda^{\rm{Com}}_{I_k} = \lambda^{\rm{Piv}}_{I_k} = 1$, so
    $$\mu_\ell - \mu_{\rm{unif}} = 2^{-\binom{k}{2}} \sum_{G \neq I_k} (\lambda_G^{\rm{Com}})^{m_1} (\lambda_G^{\rm{Piv}})^{m_2} \chi_G.$$
    We now decompose the sum into the families $\cG_{k,2r}$ for all $r \geq 0$, noting that $\cG_{k,0} = \{I_k\}$. Using the eigenvalue bounds in \Cref{lem:eigenvalue_rank} and the counting bounds of \Cref{thm:graphrank_counting}, we have the following for any $H \in \cG_k$:
    \begin{align*}
        |\mu_\ell(H) - \mu_{\rm{unif}}(H)| &\leq 2^{-\binom{k}{2}} \sum_{G \neq I_k} |\lambda_G^{\rm{Com}}|^{m_1}\cdot |\lambda_G^{\rm{Piv}}|^{m_2}\cdot |\chi_G(H)|\\
        &\leq 2^{-{k\choose 2}}\sum_{r=1}^{\floor{k/2}} \sum_{G\in \cG_{k,2r}}2^{-rm_1}2^{-2rm_2}\\
        &\leq 2^{-{k\choose 2}}\sum_{r=1}^\infty2^{2rk-2}2^{-rm}\\
        &\leq 2^{-{k\choose 2}+2k-m},
    \end{align*}
    where we use that $m> 2k$ so the sum converges.
\end{proof}

We use \Cref{lem:random_walk_mixing} to show that, in a graph $G \sim \bG(n,1/2)$ with $U \subseteq V(G)$, by repeatedly complementing/pivoting at distinct vertices outside $U$, the distribution of $G\circ(x_t)_t[U]$ becomes quickly uncorrelated from the initial state $G[U]$. However, to apply \Cref{lem:random_walk_mixing}, we need that the complement/pivot steps are actually independent; this relies on the initial graph being uniformly random. In fact, we need that the steps of the random walk are independent even when we choose the sequence $(x_t)_t$ based on $G - U$. Towards that end, we prove the following technical lemma.

\begin{lemma}\label{lem:complementing_independence}
    Let $W\subseteq \hat V \subseteq [n]$ and $U = [n] \setminus \hat V$, and fix a graph $\hat{G}$ with vertex set $\hat V$. Let $v_1, \dots, v_t\in W$ be a (possibly repeating) sequence of vertices depending on $\hat{G}$. Let $G$ be a random graph with vertex set $[n]$ drawn from $\bG(n,1/2)$ conditioned on $G[\hat V] = \hat{G}$. Partition $\binom{[n]}{2}$ into $E_1 = \binom{\hat V}{2} \cup [W, U]$ and $E_2 = \binom{U}{2} \cup [\hat V \setminus W, U]$. Let $G' = G * (v_i)_{i=1}^t$. Then $E(G') \cap E_2$ is uniformly random and independent of $E(G) \cap E_1$.
\end{lemma}
\begin{proof}
    Fix a subset $A$ of $E_1$. It suffices to show that conditioned the event that $E(G) \cap E_1 = A$, the distribution of $E(G') \cap E_2$ is uniformly random. We proceed by induction on $t$. If $t=0$ then $G'=G$, so by definition of $\bG(n,1/2)$, $E(G') \cap E_2$ is uniformly random.

Suppose $t\geq 1$ and let $G''=G*(v_i)_{i=1}^{t-1}$; by the induction hypothesis, $E(G'')\cap E_2$ is uniformly random conditioned on $E(G) \cap E_1$. Observe that $G'=G''*v_t$ is obtained by the involutive map $H\mapsto H\symdif K_{T}$, where $T=N_{G''}(v_t)$ is determined exactly by $E(G) \cap E_1$. This induces an involutive map $E(G'')\cap E_2\mapsto E(G')\cap E_2$, and therefore preserves the uniform measure.
\end{proof}
\begin{remark}
We note that because the above proposition holds for every graph $\hat{G}$, it also holds if we let $\hat{G}$ be uniformly random, and thus it holds when $G \sim \bG(n,1/2)$ and $v_1, \dots, v_t$ is a random sequence depending on $G[\hat V]$.
\end{remark}

Observe that \cref{lem:complementing_independence} implies that conditioning on $G - U$ and some previous mixing steps $x_1,\ldots,x_{t-1}\subseteq W$, as long as $x_t$ is disjoint from $W$ and $U$, we have that both $G[U]$ and the edges from $x_t$ to $U$ are uniformly random. This implies that each step of our random walk is independent of the previous steps. With this, we conclude the final random walk result.

\begin{lemma}\label{lem:complementing_mixing}
Let $G\sim \bG(n + k,1/2)$ with partition $V(G)=\hat V\disjcup U$, $|U|=k$. Let $\hat G$ be a graph on vertex set $\hat V$ and let $(x_i)_{i=1}^\ell$ be a sequence of disjoint singletons and pairs from $\hat V$ such that $\hat G\circ (x_i)_{i=1}^\ell$ is well-defined and $|\bigcup_i x_i| > 2k$. Then,
\[
\Big|\PP\Big[G\circ (x_i)_{i=1}^\ell[U]=G[U]\ \Big|\ G[\hat V]=\hat G\Big]-2^{-{k\choose 2}}\Big|\leq 2^{2k-\binom{k}{2} - |\bigcup_i x_i|}.
\]
\end{lemma}

\begin{proof}
For times $t=0,1,\ldots, \ell$, define $G_t \coloneqq G\circ (x_i)_{i=1}^t$, and define the probability distribution $\mu_t\in [0,1]^{\cG_k}$ as the distribution of $H_t:=(G_t\symdif G)[U]$, conditioned on $G[\hat V]=\hat G$. For example, $\mu_0$ is the point mass distribution $\delta_{I_k}$.

At time $t$, $H_{t-1}$ is a function of $G[\hat V]=\hat G$ and $N_G(w)$ for $w\in \bigcup_{i=1}^{t-1}x_i$. If $x_t=\{v\}$, then $H_t=H_{t-1}\symdif K_S$ where $S=N_{G_{t-1}}(v)\cap U$. Similarly, if $x_t=\{v,v'\}$, then $H_t=H_{t-1}\symdif K^{\triangle}_{S,T}$, where $S=N_{G_{t-1}}(v)\cap U$ and $T=N_{G_{t-1}}(v')\cap U$. By \cref{lem:complementing_independence} (applied with $W=\bigcup_{i<t}x_i$), we have that $S$ (resp. $T$) is a uniformly random subset of $U$, independent of $G[\hat V]$ and $N_G(u)$ for $u\in \bigcup_{i<t}x_i$ and thus $S$ (resp. $T$) is independent of $H_{t-1}$. Then by definition, if $|x_t|=1$ then we have exactly $\mu_t=\rm{Com}(\mu_{t-1})$, and if $|x_t|=2$ then $\mu_t=\rm{Piv}(\mu_{t-1})$. We then apply \Cref{lem:random_walk_mixing} (noting that the $x_i$ are disjoint so $m = |\bigcup_i x_i|$) to conclude
\begin{align*}
    \Big|\PP\Big[(G_\ell\triangle G)[U]=I_k\ \Big|\ G[\hat V]=\hat G\Big]-2^{-{k\choose 2}}\Big| = |\mu_\ell(I_k) - \mu_{\rm{unif}}(I_k)|\ 
    &\leq\ 2^{2k-\binom{k}{2} - |\bigcup_i x_i|},
\end{align*}
as desired.
\end{proof}

\subsection{Proof of \texorpdfstring{\cref{thm:vm_universal}}{Theorem 1.1}}

We now use the reordering \cref{lem:minor-reordering} to show that when $J,J'$ are sufficiently different, the graphs $G_J$ and $G_{J'}$ are sufficiently uncorrelated by applying the random walk described in \cref{section:randomwalks}.

\begin{lemma}\label{lem:JJ' are mixed} Let $G\sim \bG(n+k,1/2)$, $U\subseteq V(G)$ with $|U|=k$, and let $J,J'\subseteq [n]$ with $|J \triangle J'| > 2k$. Defining $G_J,G_{J'}$ as in Equation \eqref{eq:GJ defn}, we have
\[
\big|\PP[G_J[U]=G_{J'}[U]]-2^{-{k\choose 2}}\big|\leq 2^{2k - \binom{k}{2} -|J\symdif J'|}
\]
\end{lemma}
\begin{proof}
Let $J=\{i_1<\ldots<i_\ell\}$ and $J'=\{j_1<\ldots<j_{\ell'}\}$. Observe that
\begin{equation}\label{eq:GJ'fromGJ}
G_{J'}=G_J*v_{i_{\ell}}*\cdots *v_{i_2}*v_{i_1}*v_{j_1}*v_{j_2}*\cdots *v_{j_{\ell'}}.
\end{equation}
Fix $\hat G\in \cG_n$, and condition on the event $G[\hat V]=\hat G$, which depends only on the edges inside $\hat V$. By the reordering \cref{lem:minor-reordering}, there is a sequence $(x_i)_{i=1}^r$ (depending on $\hat G$) with the three properties from that lemma such that $G_{J'}[U]=G_J\circ (x_i)_{i=1}^r[U]$. By property (III), we note that $|\bigcup_ix_i|\geq|J\symdif J'|$, since a vertex $v_i$ appears exactly once in Equation \eqref{eq:GJ'fromGJ} if $i\in J\symdif J'$. Let $\hat G_J=\hat G*(v_{i_t})_{t=1}^{\ell}$, so $G[\hat V]=\hat G$ if and only if $G_J[\hat V]=\hat G_J$. Then by \cref{lem:complementing_mixing},
\begin{align*}
\Big|\PP\Big[G_{J'}[U]=G_J[U]\ \Big|\ G[\hat V]=\hat G\Big]-2^{-{k\choose 2}}\Big|&=
\Big|\PP\Big[G_J\circ (x_i)_{i=1}^r[U]=G_J[U]\ \Big|\ G_J[\hat V]=\hat G_J\Big]-2^{-{k\choose 2}}\Big|\\
&\leq 2^{2k - \binom{k}{2} - |\bigcup_{i\leq r}x_i|}\\
&\leq 2^{2k - \binom{k}{2} -|J\symdif J'|},
\end{align*}
as desired. We apply the law of total probability over all $\hat G\in \cG_n$ to obtain the unconditioned bound.
\end{proof}

Before we complete the second moment method calculation, we need one more fact about the independence implicit in our construction.

\begin{proposition}\label{prop:symdiffindep}Let $G\sim \bG(n+k,1/2)$, $U\subseteq V(G)$ with $|U|=k$, and let $J,J'\subseteq [n]$. Defining $G_J,G_{J'}$ as in Equation \eqref{eq:GJ defn}, the graphs $G_J[U]$ and $(G_J\symdif G_{J'})[U]$ are independent.
\end{proposition}

\begin{proof} We use a different but equivalent sampling procedure to obtain $G_J$ and $G_{J'}$. First, sample $G_J\sim \bG(n,1/2)$ uniformly at random, which is the correct distribution by \cref{prop:complementation preserves random}. In particular, the sets $E(G_J[U])$ and $E(G_J)\setminus E(G_J[U])$ are independent. Next, compute $G_{J'}$ directly using Equation \eqref{eq:GJ'fromGJ}. Observe that the set $E(G_J\symdif G_{J'})$ is a function only of $E(G_J)\setminus E(G_J[U])$ since this procedure for obtaining $G_{J'}$ depends on the edges incident to the $v_i$'s. Thus the graphs $G_J[U]$ and $(G_J\symdif G_{J'})[U]$ are independent, as desired.
\end{proof}

We now perform the second moment calculation.

\begin{lemma}\label{lem:vm_high_probability} Let $G\sim \bG(n+k,1/2)$. Fix a vertex set $U\subseteq V(G)$ with $|U|=k$ and a graph $H$ on $U$. Then $H$ is a vertex-minor of $G$ with probability at least $1-2^{2k}(3/4)^n - 2^{-n + \binom{k}{2} + 2k\log_2n}$.
\end{lemma}
\begin{proof}
For all $J\subseteq [n]$, construct $G_J$ as in \eqref{eq:GJ defn}. Let $X_J$ be the indicator that $G_J[U] = H$, and let $X = \sum_{J \subseteq [n]} X_J$. Recall that $\Ex[X]=2^{n-{k\choose 2}}$ by linearity of expectation and \cref{prop:complementation preserves random}. We compute
\begin{align*}
\var[X]&=\sum_{J,J'\subseteq [n]}\Ex[X_JX_{J'}]-\Ex[X_J]\Ex[X_{J'}]\\
&=\sum_{J,J'\subseteq [n]}\PP[G_J[U]=H\wedge G_J[U]=G_{J'}[U]]-\PP[G_J[U]=H]\PP[G_{J'}[U]=H]\\
&=\sum_{J,J'\subseteq [n]} \PP[G_J[U]=H] \PP[G_J[U]=G_{J'}[U]]-\PP[G_J[U]=H]\PP[G_{J'}[U]=H],
\end{align*}
where we use the independence from \cref{prop:symdiffindep}. Recall that $G_J[U],G_{J'}[U]$ are both uniformly random, so by \cref{lem:JJ' are mixed} we have
\begin{align*}
\var[X]&=\sum_{J\subseteq [n]}2^{-{k\choose 2}} \sum_{J'\subseteq [n]}\big(\PP[G_J[U]=G_{J'}[U]]-2^{-{k\choose 2}}\big)\\
&= \sum_{J\subseteq [n]}2^{-{k\choose 2}} \sum_{J' \triangle J\subseteq [n]}\big(\PP[G_J[U]=G_{J'}[U]]-2^{-{k\choose 2}}\big)\\
&\leq \sum_{J\subseteq [n]}2^{-{k\choose 2}} \left(\sum_{\substack{J' \symdif J\subseteq [n]\\ |J' \symdif J| > 2k}}2^{2k - \binom{k}{2} - |J' \symdif J|} + \sum_{\substack{J' \symdif J\subseteq [n]\\ |J' \symdif J| \leq 2k}}1\right)\\
&= \sum_{J \subseteq [n]} 2^{-\binom{k}{2}} \left(\sum_{i=2k+1}^n \binom{n}{i}2^{2k-\binom{k}{2} - i} + \sum_{i=0}^{2k} \binom{n}{i}\right)\\
& \leq \sum_{J \subseteq [n]} 2^{-\binom{k}{2}} \left(2^{2k - \binom{k}{2}}\sum_{i=0}^n \binom{n}{i}2^{-i} + n^{2k}\right)
\\&\leq \sum_{J \subseteq [n]} 2^{-\binom{k}{2}} \left(2^{2k-\binom{k}{2}}(3/2)^n + n^{2k}\right)\\
&= 2^{n - \binom{k}{2}}\left(2^{2k-\binom{k}{2}}(3/2)^n + n^{2k}\right).
\end{align*}
(To obtain the bound $\sum_{i=0}^{2k} \binom{n}{i}\leq  n^{2k}$, consider for example the surjection $f:[n]^{2k} \rightarrow \binom{n}{\leq 2k}$ given by $f(x_1, \dots, x_{2k}) = \{x_1, \dots, x_{2k}\}$, except $f(2k, 2k-1, \ldots, 1)=\emptyset$.)
In conclusion, by Chebyshev's inequality (\cref{lem:chebyshev}) we have
\begin{align*}
    \PP[X = 0] &\leq \frac{\var[X]}{\Ex[X]^2}
    \leq \frac{2^{n - \binom{k}{2}}\left(2^{2k-\binom{k}{2}}(3/2)^n + n^{2k}\right)}{2^{2n - 2\binom{k}{2}}}
    = 2^{2k}(3/4)^n + 2^{-n+\binom{k}{2}}n^{2k}
\end{align*}
as desired.
\end{proof}

We are now ready to prove our main result via the union bound over all choices of $U \subseteq V(G)$ and graphs $H$ with vertex set $U$. We first restate the theorem for convenience.

\vmunivthm*

\begin{proof}
Let $G\sim\bG(n,p)$, where $n = (1 + c) \frac{1}{2\log_2(4/3)}k^2$. Here we allow $c=c(k)$ to be increasing with $k$, but it is lower bounded by some positive constant. 
By \cref{lem:vm_high_probability}, the probability that a given graph $H$ on $k$ given vertices of $V(G)$ is not a vertex-minor of $G$ is at most $2^{2k}(3/4)^{n-k} + 2^{-n+k+\binom{k}{2}}(n - k)^{2k}=2^{2k}(3/4)^{(1+o(1))n}$. 

By the union bound, we have
\begin{align}\label{eq:vm-final-bound}
    \PP[G\text{ is not $k$-vertex-minor universal}] &\leq {n\choose k}2^{{k\choose 2}}\cdot \left(2^{2k}(3/4)^{(1+o(1))n}\right)\\
    &\leq \exp_2\left(k\log_2 n + \frac{k^2}{2} + 2k - (1 + o(1))\log_2(4/3)n\right)\notag\\
    &= \exp_2\left(O(k\log(ck)) + \frac{k^2}{2} - (1 + o(1))\frac{1 + c}{2}k^2\right)\notag\\
    &= \exp_2\left(O(k\log(ck))-\left(1 + o(1)\right)ck^2/2\right).\notag
\end{align}
Thus, whenever $n = (1+c)\frac{1}{2\log_2(4/3)}k^2$, we have
$\PP[G\text{ is $k$-vm-universal}]\geq 1-2^{-(1+o(1))c k^2/2},$
as desired. 
\end{proof}

\subsection{Proof of reordering \texorpdfstring{\cref{lem:minor-reordering}}{Lemma 2.3}}\label{section:proofofreorderinglemma}

First, we restate the vertex-minor reordering lemma:

\minorreorderinglemma*

\noindent First, we prove a slightly stronger lemma by induction which gives a sequence that depends on the host graph $G$.

\begin{lemma}\label{lem:inductive_reordering}
For any graph $\gn$ with $V(G) = \hat V \disjcup U$ and $v_1,\ldots, v_\ell\in \hat V$, and for any special vertex $v \in \hat V$, there exists a sequence $(y_i)_{i=1}^r$ satisfying (I) and (II) for $\gn$ such that one of the following holds:
    \begin{enumerate}
        \item\label{itm:v not present} we have $v \not\in \bigcup_i y_i$,
        \item\label{itm:v first} we have $v \in y_1$, or
        \item\label{itm:v second} we have $y_2 = \{v\}$ and $y_1 = \{w\}$ for some $w \in \hat V$ with $vw \in E(\gn)$.
    \end{enumerate}
\end{lemma}
\begin{proof}
Let $H = \gn * v_1 * \cdots * v_\ell[U]$, so $H$ is a vertex-minor of $\gn$. We proceed by induction on $n:=|V(G)|$; the base case $n=1$ is trivial. Let $v \in \hat V$ be arbitrary.

By \cref{thm:structure of minors}, $H$ is a vertex-minor of $\gn - v$, $\gn*v - v$, or $\gn \times vw - v$ for some $w$ adjacent to $v$. In the first case, we apply the induction hypothesis to $\gn - v$, and clearly the resulting sequence $(y_i)_i$ satisfies item~\ref{itm:v not present}. In the second case, we apply the induction hypothesis to $\gn * v - v$ to obtain the sequence $(y_i)_{i \geq 2}$. Then clearly setting $y_1 = \{v\}$ satisfies item~\ref{itm:v first}.

In the final case, we apply the induction hypothesis to $\gn \times vw - v$ for some $w$ adjacent to $v$ such that $w$ is the special vertex to obtain the sequence $(y_i)_{i \geq 2}$. We then split into cases based on the location of $w$. If $w \not\in \bigcup_{i \geq 2} y_i$, then setting $y_1 = \{v,w\}$ gives the desired sequence $(y_i)_{i\geq 1}$ satisfying item~\ref{itm:v first}. If $y_2 = \{w\}$, then note $\gn \times vw * w = \gn *w * v$, and so replacing $y_2$ with $\{v\}$ and setting $y_1 = \{w\}$ gives the desired sequence $(y_i)_{i\geq 1}$ satisfying item~\ref{itm:v second}. If $y_2 = \{w, u\}$, then by \cref{lem:cancel pivots}, $\gn \times vw \times wu = \gn \times vu$, and so replacing $y_2$ with $\{v, u\}$ yields the desired sequence $(y_i)_{i \geq 2}$ satisfying item~\ref{itm:v first}. Finally, if $y_2 = \{u\}, y_3 = \{w\}$ with $uw\in E(G\times vw)$, then we use the fact that $\gn \times vw * u * w = \gn \times vw \times uw * u = \gn \times vu * u = \gn * u * v$, and so by replacing $y_2, y_3$ with $\{u\}, \{v\}$ respectively, we get the desired sequence $(y_i)_{i\geq 2}$ satisfying item~\ref{itm:v second}.
\end{proof}

\begin{proof}[Proof of \cref{lem:minor-reordering}]
Let $\hat V$ be a set, and fix a sequence $v_1,\ldots, v_\ell\in \hat V$ and a graph $\hat G$ on vertex set $\hat V$. To show that we can pick $(x_i)_i$ to depend only on $\hat G$, we apply \cref{lem:inductive_reordering} to a universal ``gadget" graph defined as follows. Let $\gadget$ have vertex set $\hat V\cup [3\cdot 2^{3|\hat V|}]$ by setting $\gadget[\hat V]=\hat G$, placing $2^{3|\hat V|}$ disjoint copies of the 3 vertex path $P_3$, and adding edges from $\hat V$ to the copies of $P_3$ such that all of the $2^{3|\hat V|}$ possible sets of edges occur between $\hat V$ and a copy of $P_3$.
This has the property that for any graph $G$ on vertex set $\hat V\cup \{u_1,u_2\}$ with $G[\hat V]=\hat G$, there exist $j_1,j_2$ such that $\gadget[\hat V\cup\{j_1,j_2\}]$ is isomorphic to $G$ with the isomorphism $\hat V\mapsto \hat V$ and $u_i\mapsto j_i$. Such an isomorphism can be found by taking the copy of $P_3$ with the correct adjacency to $\hat V$ and letting $j_1, j_2$ be either adjacent or nonadjacent vertices in the $P_3$.

\begin{numclaim}{\ref{lem:minor-reordering}.1}For all $k\geq 2$, if the sequence $(x_i)_{i=1}^r$ satisfies properties (I) and (II) for the gadget $\gadget$, then it satisfies property (III).
\end{numclaim}
\begin{subproof}
Suppose $v\in \hat V$ appears exactly once in $\{v_i\}_i$ in position $i^*$, and suppose $(x_i)_i$ does not contain $v$. Expanding out $\circ (x_i)_i=*(w_i)_{i=1}^m$ as a sequence of complementations for $w_i\in \hat V\setminus\{v\}$, for $\gadget'=\gadget*v_1*\cdots *v_{i^*-1}$, we have
\begin{equation}\label{eq:gadget}\gadget'*v=\gadget'*v_{i^*-1}*\cdots *v_1*w_1*\cdots *w_m*v_\ell*\cdots *v_{i^*+1}.\end{equation}
We show this is impossible. Let $\RHS$ denote the right-hand side of \eqref{eq:gadget}. Choose vertices $j_1,j_2>n$ such that $j_1j_2\notin E(\gadget)$ and $N_{\gadget}(j_1)\cap \hat V=N_{\gadget}(j_2)\cap \hat V=\{v\}$; these vertices exist by definition of $\gadget$. Then we note $j_1j_2 \not\in E(\gadget')$ and $N_{\gadget'}(j_1) = N_{\gadget}(j_1)$, $N_{\gadget'}(j_2) = N_{\gadget}(j_2)$.
Then $j_1j_2\in E(\gadget'*v)$ but  $j_1j_2\notin E(\RHS)$, because no complementation step can add an edge incident to $j_1,j_2$. This contradicts \eqref{eq:gadget}, so $(x_i)_i$ contains $v$.
\end{subproof}

We now show that if a sequence satisfies \cref{lem:inductive_reordering} for $\gadget$, then it satisfies \cref{lem:minor-reordering}. Let $(x_i)_{i=1}^r$ be the sequence given by \cref{lem:inductive_reordering} for the gadget $\gadget$. Let $G$ be an arbitrary graph with $\hat V\subseteq V(G)$ and $G[\hat V]=\hat G$. We need to show that the sequence satisfies condition (I) for $G$. For this it suffices to show that for all $u_1,u_2\in V(G)\setminus \hat V$, we have $u_1u_2\in E(G*(v_i)_{i=1}^\ell)$ if and only if $u_1u_2\in E(G\circ (x_i)_{i=1}^j)$. Indeed, letting $G'=G[\hat V\cup \{u_1,u_2\}]$, there exist $j_1,j_2\in V(\gadget)$ such that $\gadget[\hat V\cup \{j_1,j_2\}]$ is isomorphic to $G$ with isomorphism $\hat V\mapsto \hat V$ and $u_i\mapsto j_i$. In particular, this means $G'*(v_i)_i-\hat V=G'\circ(x_i)_i-\hat V$, so we conclude that $u_1u_2\in E(G'*(v_i)_i)$ if and only if $u_1u_2\in E(G'*(v_i)_i)$, as desired. Thus we conclude $G*(v_i)_i-\hat V=G\circ(x_i)_i-\hat V$, so all three properties hold.
\end{proof}

\section{Universality for pivot-minors and binary matroids}\label{section:Universality for pivot-minors and binary matroids}

Although pivot-minors are not as relevant for quantum information theory, they are much-studied both in a purely combinatorial setting as well as in the context of matroids. 
For each binary matroid, we can associate a bipartite graph called the \textit{fundamental graph}. It turns out that pivot-minors on this bipartite graph correspond exactly to minors of this binary matroid; see Section~\ref{section:pivot-minors and matroids} for more details. 

Clearly a pivot-minor of $G$ is also a vertex-minor, so pivot-minors on general graphs serve as a common generalization of both vertex-minors and minors in binary matroids. For example, Geelen and Oum~\cite{geelen2009circle} gave a forbidden pivot-minor characterization of circle graphs (which are analogogus to planar graphs for vertex-minors~\cite{de1981local}), which implies both Bouchet's forbidden vertex-minor characterization of circle graphs \cite{bouchet1994circle} and Kuratowski's Theorem. As another example, Campbell et. al.~\cite{Cam+} proved the Erd\H{o}s-P\'{o}sa property for circle graphs as vertex-minors, and were able to extend their techniques to prove the Erd\H{o}s-P\'{o}sa property for planar multigraphs in binary matroids. In a similar way, our \cref{thm:ER_pm_universal} generalizes \cref{thm:vm_universal} to pivot-minors, but with a worse constant factor.

\cref{thm:ER_pm_universal} does not imply an analogous result for binary matroids because the uniformly random graph is almost never bipartite. However, we are able to use similar methods to prove a bipartite version of pivot-minor universality, which we now describe.

Observe that an equivalent way to view the pivot operation is that for $vu \in E(G)$, $G \times vu$ is the graph obtained from $G$ by taking the symmetric difference of $G$ with the complete tripartite graph with parts $N(v) \cap N(u),\  N(v) \setminus (N(u) \cup \{u\}),\  N(u) \setminus (N(v) \cup \{v\})$, and then swapping the labels of $u$ and $v$.
An important consequence of this description is the following.
\begin{observation}
    If $G$ is bipartite and $vu \in E(G)$, then $G \times vu$ is bipartite.
\end{observation}
To work with the strongest possible notion of bipartite pivot-minor universality, which is necessary for the equivalence to labeled matroid-minor universality, we work in the setting of ordered bipartite graphs.
\begin{definition}
    An \emph{ordered bipartite graph} is a triple $G = (L,R,E)$ where $(L \cup R, E)$ is a bipartite graph with bipartition $L,R$. We define $L(G),R(G),E(G),$ and $V(G) = L(G) \cup R(G)$ in the natural way. We define the pivot operation on ordered bipartite graphs as follows. For $uv\in E(G)$, we define $G\times uv\coloneqq (L\symdif \{u,v\},R\symdif\{u,v\}, E'),$ where $E'$ is the edge set of the unordered bipartite graph $(L\cup R, E)\times uv$. For any ordered bipartite graph $H = (U_1, U_2, E_H)$ with $U_1 \cup U_2 \subseteq V(G)$, we say that $H$ is a (labeled) pivot-minor of $G$ (in the ordered sense) if $H$ can be obtained exactly by a sequence of vertex deletions and pivots. 
\end{definition}
Note that this is strictly stronger than only requiring that graph $(U_1\cup U_2,E_H)$ is a (labeled) pivot-minor of $(L\cup R,E)$, since we require the partitions of $H$ and $G$ to be aligned.

Moreover, recall that because we preserve vertex labels, graphs may be isomorphic without one being a pivot-minor of the other.

\begin{definition}
An ordered bipartite graph $G$ is \emph{bipartite $k$-pivot-minor universal} if any ordered bipartite graph $H$ on any $k$ vertices of $V(G)$ is a pivot-minor of $G$.
\end{definition}

Let $\bG(\ell,r,1/2)$ be the uniform distribution on ordered bipartite graphs on parts $L,R$ with $|L|=\ell$ and $|R|=r$, where each edge is included independently with probability $1/2$. More generally, for disjoint (possibly random) sets $L,R$, we define $\bG(L,R,1/2)$ to be a uniformly random ordered bipartite graph with parts $L,R$. 

\begin{theorem}
\label{thm:bip_pm_universal} Let $C=\frac{1}{2\log_2(8/7)}+\frac 14\approx 2.85$. 
For all $c>0$, if  $\ell\geq r\geq (1+c)Ck^2$ and $\ell \leq 2^{o(\sqrt r)}$, then $G\sim \bG(\ell, r,1/2)$ is bipartite $k$-pivot-minor universal with probability $1-2^{-(1+o(1))ck^2/4}$.
\end{theorem}

This bipartite result allows us to move to matroid minor universality.

\begin{definition}
A rank-$r$ matroid $M$ on ground set $[n]$ is \emph{$k$-minor-universal} if for any subset $U\subseteq[n]$ with $|U|=k$ and any matroid $N$ with ground set $U$, $N$ is a (labeled) minor of $M$.
\end{definition}

We now restate \cref{thm:matroid_universal} with all the relevant notions defined.

\matroidUniversal*

\cref{thm:matroid_universal} is a consequence of the following stronger statement. The reduction follows from the fact that the rank of a uniformly random $n$-element binary matroid is $(1+o(1))n/2$ with very high probability (see \cref{prop:matroid_rank_concentration}).

\begin{theorem}\label{thm:stronger_matroid_universal}
Let $C=\frac{1}{2\log_2(8/7)}+\frac 14$ and $c>0$. Let $M$ be a uniformly random rank-$r$ binary matroid on ground set $[n]$. If $r,n-r\geq (1+c)Ck^2$ and $\max(r,n-r)\leq 2^{o(\sqrt{\min(r,n-r)})}$, then $M$ is $k$-minor universal with probability at least $1-2^{-(1+o(1))ck^2/4}$.
\end{theorem}

We contrast this with a result of Altschuler and Yang~\cite{altschuler2017inclusion}, which implies that for a fixed binary matroid $N$ with rank $r'$ on $k$ elements, if $M$ is a binary matroid represented by a uniformly random $r\times n$ binary matrix with $r\geq r'$ and $n=n(r)$, then as $(n(r) - r) \rightarrow \infty$, $M$ contains $N$ with high probability. In particular, their quantitative results \cite[Theorems 6,8]{altschuler2017inclusion} together with \cref{prop:matrix_vs_matroid_dist} imply that if $M$ is a uniformly random rank-$r$ binary matroid on ground set $[n]$ with $r\geq r'$, and $n\geq k$, then
\begin{align*}
\PP[N\text{ is isomorphic to a minor of }M]&\geq 1-\exp\left(-\Omega\Big(2^{-r'(k-r')}\frac{n-r}{k}\Big)\right)\\
&\geq 1-\exp\left(-\Omega(2^{-k^2/4}(n-r)/k)\right).
\end{align*}
In particular, as $n\to\infty$ this converges to $1$, and more specifically $M$ contains every matroid on $k$ elements as an \textit{unlabeled} minor with high probability as long as $(n-r)\geq 2^{k^2/4+\Omega(\log k)}$. Our result operates in the regime $r,(n-r)\geq \Omega(k^2)$ and obtains each \textit{labeled} matroid $N$ as long as $r,(n-r)$ are not too different.

In the deterministic setting, the growth rate theorem of Geelen, Kung and Whittle~\cite{geelen2009growth} implies that for every fixed $k$, there exists a (fast-growing) constant $C_k$ such that every simple binary matroid $M$ of rank $r$ with $n\geq C_kr^2$ elements contains every $k$-element matroid as an unlabeled minor; a matroid is \textit{simple} if there are no loops or parallel elements. We note that if $r \leq \log_2(C_k)$, the statement is vacuous as no such simple $M$ exists, so this theorem implicitly forces $r$ to be large as a function of $k$. Cooper, Frieze, and Pegden~\cite{cooper2019minors} improved this to a linear bound $n\geq C_k'r$ when $M$ is the graphic matroid of a uniform hypergraph with $n$ i.i.d. random edges. However, with only astronomical bounds on $C_k$ and $C_k'$, these results are incomparable with ours. Other discussion of various random matroid models can be found in \cite{evolutionrandomrepmatroids,GaoNelson2026,KellyOxley1984,Kordecki_1996,Nelson2018}, and the references therein.

In \cref{section:pivot-minors and matroids} we provide more detailed background about matroids before proving the reduction from \cref{thm:stronger_matroid_universal} to \cref{thm:bip_pm_universal}. Then, in \cref{section:bipartite-pivot-minor-universality} we prove \cref{thm:bip_pm_universal} using similar techniques to the proof of \cref{thm:vm_universal}.

We remark that the proof of \cref{thm:ER_pm_universal} is similar to the proof of \cref{thm:vm_universal} but using some ideas from the proof of \cref{thm:bip_pm_universal}. Since no additional ideas are involved, we omit the proof of \cref{thm:ER_pm_universal} in this manuscript.

\subsection{Reduction from binary matroids to bipartite pivot-minors}\label{section:pivot-minors and matroids}

We now discuss the connection between pivot-minors in ordered bipartite graphs and minors of binary matroids, reducing \cref{thm:stronger_matroid_universal} to \cref{thm:bip_pm_universal}. More details on this relationship can be found in \cite{Oum05}. We assume familiarity with basic matroid theory, see Oxley's book~\cite{oxley2006matroid} for reference. All matroid minors will be labeled (not equivalent up to isomorphism). Let $G = (L,R,E)$ be an ordered bipartite graph. Let $A=\Adj(G) \in \{0,1\}^{R \times L}$ be the bipartite adjacency matrix between parts $L$ and $R$. That is, with $A_{u,v} = 1$ if and only if $u \in R$ is adjacent to $v \in L$. To each ordered bipartite graph $G = (L,R,E)$ we associate the binary matroid $M(G) = M(L,R,E)$ on ground set $L \disjcup R$ represented by the matrix $[I_R\ A] \in \{0,1\}^{R \times V(G)}$ where $I_R$ is the $R \times R$ identity matrix. We then say $G$ is the \textit{fundamental graph} of $M(G)$ with respect to the basis $R$. 

We note that any rank $|R|$ binary matroid can be represented by a matrix of the form $[I_R\ A]$. Equivalently, we can define a fundamental graph of a binary matroid $M$ as follows. Fix a basis $R$ of $M$, and let $L$ be the remainder of the ground set. Then the fundamental graph of $M$ with respect to $R$ is the ordered bipartite graph $(L,R,E)$ where $v \in L$ is adjacent to exactly the elements of $R$ which appear in its fundamental circuit. We now review the relation between pivot-minors of the fundamental graph and the minors of $M$.

\begin{proposition}[{\cite[Proposition 3.3]{Oum05}}]\label{prop:fundamental_graph_equivalence}
    Let $G = (L,R,E)$ be an ordered bipartite graph and let $M(G) = M(L,R,E)$. Then
    \begin{enumerate}
        \item $M(R,L,E) = M(L,R,E)^*$
        \item if $uv \in E(G)$, $M(G \times uv) = M(G)$
        \item $M(G - v) = \begin{cases}
            M \setminus v & v \in L\\
            M / v & v \in R
        \end{cases}$
    \end{enumerate}
\end{proposition}

That is, flipping the order of the bipartition takes the dual of the matroid, pivoting on an edge exchanges the two elements in the chosen basis but does not change the matroid, and deleting a vertex either contracts or deletes the element in the matroid. This means, for instance, to contract a non-loop element in $L$, you can pivot on an incident edge to move it to $R$, and then delete the element. This implies the following.

\begin{corollary}[{\cite[Corollary 3.6]{Oum05}}]\label{cor:pm_equiv_to_matroid_minor}
    Let $M, N$ be binary matroids with fundamental graphs $G,H$ respectively. If $N$ is a minor of $M$, then $H$ is a pivot-minor of $G$.

    Conversely, let $G = (L,R,E)$ and $H=(L',R',E')$ be ordered bipartite graphs. If $H$ is a pivot-minor of $G$ (in the ordered sense), then $M(L',R',E')$ is a minor of $M(L,R,E)$.
\end{corollary}

This implies that a matroid $M$ is $k$-minor universal if and only if its fundamental graph $G$ is bipartite $k$-pivot-minor universal. This is almost sufficient to reduce \cref{thm:stronger_matroid_universal} to \cref{thm:bip_pm_universal}, but with one technical issue: the uniform distribution on the set $\cM_{r,n}$ of rank-$r$ binary matroids on ground set $[n]$ (henceforth $\mu_{r,n}$) is not equivalent to the distribution $\bG(\ell,r,1/2)$ on their fundamental graphs. The remainder of this section will show that these two distributions are actually very close.

For a matroid $M$, let $b(M)$ be the number of bases of $M$. The following proposition relates the distributions $\bG(\ell,r,1/2)$ and $\mu_{r,n}$ by the proportionality $b(M)$.

\begin{proposition}\label{prop:graph_vs_matroid_dist} Let $R$ be a uniformly random $r$-element subset of $[n]$ and $L=[n]\setminus R$, and let $G\sim \bG(L,R,1/2)$. Then for any matroid $M\in\cM_{r,n}$, we have $\PP[M(G)=M]=b(M)/z$, with normalizing constant $z=\sum_{M'\in \cM_{r,n}}b(M') = {n\choose r}2^{r(n-r)}$.
\end{proposition}
\begin{proof}
Observe that each matroid has $b(M)$ labeled fundamental graphs, one for each basis $B\subseteq [n]$. Then $G$ is the unique fundamental graph of $M$ with respect to $B$ precisely if $R=B$ and $E(G)$ is chosen correctly, with probability ${n\choose r}^{-1}2^{-r(n-r)}$. Thus, $\PP[M(G)=M]=b(M)/{n\choose r}2^{r(n-r)}$.
\end{proof}

To show that the distributions $\bG(\ell,r,1/2)$ and $\mu_{r,n}$ are not too far apart, it now suffices to show that $b(M)$ is concentrated in the latter distribution. First, we show that $\mu_{r,n}$ is close to the ``random representable matroid" model of Kelly and Oxley \cite{KellyOxley1984}, and also studied in \cite{evolutionrandomrepmatroids,altschuler2017inclusion}. For $A\in \bF_2^{r\times [n]}$, let $M(A)$ be the labeled matroid represented by $A$. Note that $M(A)\in \cM_{r,n}$ if and only if $A$ is full rank. Finally, it is easy to show \cite[Prop.\ 6.4.1]{oxley2006matroid} that each $M\in \cM_{r,n}$ is uniquely representable up to row operations, and thus has exactly $|\mathrm{GL}_r(2)|$ representation matrices $A\in \F_2^{r\times [n]}$, where by $\mathrm{GL}_r(2)$ denotes the general linear group of $r \times r$ invertible matrices over $\F_2$. This yields the following: 

\begin{proposition}\label{prop:matrix_vs_matroid_dist}
If $A\in \bF_2^{r\times n}$ is a uniformly random binary matrix, then $\mu_{r,n}$ is precisely the conditional distribution of $M(A)$, conditioned on the event that $A$ is full rank.
\end{proposition}
This relationship is notably false for general $\F_q$ representable matroids. It is easy to check {\cite[Lemma 3.1]{KellyOxley1984}} that $A\in \bF_2^{r\times n}$ is full rank with probability $\prod_{i=1}^r(1-2^{r-n-i})\geq 1-2^{-(n-r)}$, and therefore the uniform distribution is close to the random representable matroid in total variation distance. Next, we concentrate $b(M)$ for a uniformly random or random representable matroid.

\begin{lemma}\label{lem:matroid_basis_count_concentration} 
Let $M\sim \mu_{r,n}$ be a uniformly random rank-$r$ binary matroid. Then
\[
\frac{\var[b(M)]}{\Ex[b(M)]^2}
\leq 4\Big(\frac{n+r}{2n}\Big)^r.
\]
\end{lemma}
\noindent Note, this means that in particular that $\var[b(M)]/\Ex[b(M)]^2=O(e^{-r(n-r)/2n})$.
\begin{proof}
Let $A\in \bF_2^{r\times n}$ be a uniformly random binary matrix and $X$ be the number of bases of $\bF_2^r$ formed by columns of $A$. Then by \cref{prop:matrix_vs_matroid_dist}, we have that the distribution of $b(M)$ is equal to the distribution of $X$ conditioned on $X > 0$. Setting $p=\PP[X>0]$, this implies that $\Ex[b(M)^2]=\Ex[X^2]/p$ and $\Ex[b(M)]=\Ex[X]/p$, so then
\[\frac{\var[b(M)]}{\Ex[b(M)]^2}
=\frac{\Ex[X^2]/p-(\Ex[X]/p)^2}{(\Ex[X]/p)^2}\leq 
\frac{\Ex[X^2]-\Ex[X]^2}{\Ex[X]^2}
= \frac{\var[X]}{\Ex[X]^2},
\]
and thus it suffices to bound $X$. We write $X=\sum_{S\in {[n]\choose r}}X_S$, where $X_S$ is the indicator that columns $S$ form a basis of $\bF_2^r$. Let 
$p_r\coloneqq |\rm{GL}_r(2)|/2^{r^2}=\prod_{i=1}^r(1-2^{-i})$
be the probability that a random $r\times r$ binary matrix is invertible, so $p_r>\prod_{i=1}^\infty(1-2^{-i})>1/4$ for all $r$. By linearity of expectation, $\Ex[X]=p_r{n\choose r}$. Then
\begin{align*}
\var[X]&= \sum_{S,T}\Ex[X_SX_T]-\Ex[X_S]\Ex[X_T]
= {n\choose r}p_r\sum_{T}(\PP[X_T\mid X_{[r]}]-p_r).
\end{align*}
If $|T\cap [r]|=k$, then by checking the independence of columns $T\setminus[r]$ one at a time, we have $\PP[X_T\mid X_{[r]}]=\prod_{i=k}^{r-1}(1-2^{i-r})=p_{r-k}$. Observe that
\[p_{r-k}-p_r=p_{r-k}\Big(1-\prod_{i=r-k+1}^r(1-2^{-k})\Big)\leq 2^{-r+k}-2^{-r},\]
and thus we compute
\begin{align*}
\frac{\var[X]}{\Ex[X]^2}
= \frac{1}{\Ex[X]}\sum_{k=0}^r\sum_{\substack{T\\
T\cap [r]=k}}(p_{r-k}-p_r)
&\leq \frac 1{p_r}\sum_{k=0}^r\frac{{r\choose k}{n-r\choose r-k}}{{n\choose r}}(2^{-r+k}-2^{-r})\\
&=p_r^{-1}2^{-r}\Ex_{Y\sim \rm{Hypergeom}(n,r,r)}[2^Y-1].
\end{align*}
The hypergeometric distribution $\rm{Hypergeom}(n,r,r)$ (i.e.\ sampling without replacement) is dominated in convex order by the binomial $\rm{Bin}(r,r/n)$ (implied by Heoffding \cite[Theorem 4]{hoeffding}). Since $x\mapsto 2^x$ is convex, we conclude
\[
\frac{\var[X]}{\Ex[X]^2}\leq 4\cdot 2^{-r}\Ex_{Y\sim \rm{Bin}(r,r/n)}[2^Y-1]=4\cdot 2^{-r}((1+r/n)^r-1)\leq 4 \Big(\frac{n+r}{2n}\Big)^r,
\]
using the binomial moment generating function. 
\end{proof}
\begin{remark}Note that when $n\geq r^2$, this final equation gives the bound $6\cdot 2^{-r}\frac{r^2}{n}$, which is mildly stronger than then Lemma as stated.
\end{remark}

Our final lemma leverages this concentration to finish reducing \cref{thm:stronger_matroid_universal} to \cref{thm:bip_pm_universal}. In particular, if $\bG(\ell,r,1/2)$ is $k$-pivot-minor universal with exponentially high probability and $r,(n-r)\geq k^2$, then a matroid $M_{\rm{unif}}\sim\mu_{r,n}$ is also $k$-minor universal with essentially the same certainty.

\begin{lemma} Let $M_{\rm{unif}}\sim\mu_{r,n}$ be uniformly random and $G\sim G(\ell, r,1/2)$ with $\ell=n-r$. Then for all $k$,
\[
\PP[M_{\rm{unif}}\emph{ is not $k$-minor universal}]\leq 16e^{-\min(r,n-r))/2}+2\PP[G\emph{ is not bipartite $k$-p.m.\ universal}].
\]
\end{lemma}
\begin{proof}Let $Q\subseteq \cM_{r,m}$ be the set of binary matroids which are not $k$-minor universal. Let $R\subseteq [n]$ be a uniformly random $r$-element subset and $L=[n]\setminus R$, and let $G'\sim \bG(L,R,1/2)$. By \cref{cor:pm_equiv_to_matroid_minor}, we have
\[
\PP[G\text{ is not bipartite $k$-p.m.\ universal}] = \PP[M(G)\in Q]= \PP[M(G')\in Q].
\]
where the last equality follows because the property of pivot-minor universality is closed under relabeling vertices. By \Cref{prop:graph_vs_matroid_dist}, for all $M'\in \cM_{r,n}$, we have $\PP[M(G')=M']= b(M')/z$, where 
\[z=\sum_{M'\in \cM_{r,n}}b(M')=|\cM_{r,n}|\cdot  \Ex[b(M_{\rm{unif}})]=|\cM_{r,n}|\cdot  \hat b,
\]
where $\hat b:=\Ex[b(M_{\rm{unif}})]$. We can partition the left hand side into two cases
\begin{equation}\label{eq:bipartite-to-matroid-reduction}
\PP[M_{\rm{unif}}\in Q]\leq \PP[b(M_{\rm{unif}})\leq \hat b/2]+\PP[b(M_{\rm{unif}})\geq \hat b/2\text{ and } M_{\rm{unif}}\in Q].
\end{equation}
For the first term, we apply Chebyshev's inequality (\cref{lem:chebyshev}) and \Cref{lem:matroid_basis_count_concentration}:
\[\PP[M_{\rm{unif}}\leq \hat b/2]\leq \frac{\var[b(M_{\rm{unif}})]}{(\hat b/2)^2}\leq 16\left(\frac{n+r}{2r}\right)^r\leq 16e^{-r(n-r)/2n}\leq 16e^{-\min(r,n-r)/2}.
\]
For the second term, we have
\begin{align*}
\PP[b(M_{\rm{unif}})\geq \hat b/2\cap M_{\rm{unif}}\in Q]
=\sum_{\substack{M'\in Q\\b(M')\geq \hat b/2}}\frac{1}{|\cM_{r,n}|}
\leq \sum_{\substack{M'\in Q}}\frac{2b(M')}{|\cM_{r,n}|\cdot \hat b}
&=2\sum_{\substack{M'\in Q}}\PP[M_G=M']\\
&=2\PP[M_G\in Q].
\end{align*}
Combining the two bounds in \Cref{eq:bipartite-to-matroid-reduction}, we obtain the desired conclusion.
\end{proof}

Finally, we reduce \cref{thm:matroid_universal} to \cref{thm:stronger_matroid_universal} by proving that a uniformly random binary matroid on $[n]$ has rank $r\in [\frac n2- \sqrt n,\frac n2+\sqrt n]$ with probability $1-O(2^{-n})$.

\begin{proposition}\label{prop:matroid_rank_concentration} There exists $c\in \R$ such that if $M$ is a uniformly random binary matroid on $n$ elements and $\delta>0$, then $\PP[|\mathrm{rank}(M)-n/2|\geq \delta]\leq c2^{-\delta^2}$.
\end{proposition}
\begin{proof}
It is well known that binary matroids on $[n]$ are in bijection with subspaces of $\F_2^n$, so $|\cM_{r,n}|=\big[\smat{
n\vspace{2pt}\\r}\big]_2=\Theta(2^{r(n-r)})$, where the constants are uniform in $r$. Then
\[
\PP[\rank(M)\leq n/2-\delta]=\Theta\Big(\frac{\sum_{r=0}^{n/2-\delta}2^{r(n-r)}}{\sum_{r=0}^n2^{r(n-r)}}\Big)\leq \Theta\Big(\frac{2^{(n/2-\delta)(n/2+\delta)}}{2^{\floor{n/2}\ceil{n/2}}}\Big)\leq \Theta\Big(2^{-\delta^2}\Big),
\]
where we lower bound the denominator by the middle term and upper bound the numerator by a geometric series. The other tail is symmetrical, giving the desired bound with some universal constant $c$.
\end{proof}

\subsection{Bipartite pivot-minor universality}\label{section:bipartite-pivot-minor-universality}

In this section, we modify the proof of vertex-minor universality to prove \cref{thm:bip_pm_universal}. Most of the lemmas can be adapted without much change, but there are some additional ideas necessary when we restrict to the pivot operation.
The first obstacle is that we cannot pre-fix a set of pivot operations, because it is not always possible to pivot at a pair $\{u,v\}$; this depends on the underlying graph. Thus, we define the ``attempted pivot" operation:
\[
G\apiv vw=\begin{cases}
G\times uv&\text{if }uv\in E(G)\\
G&\text{if }uv\notin E(G).
\end{cases}
\]

Let $G^*=(L^*,R^*,E^*)\sim \bG(\ell,r,1/2)$. Our goal is to show that for any subset $U$ of size $k$ and any ordered biparitite graph $H=(U_1,U_2,E_H)$ on vertex set $U$, we have that $H$ is a pivot-minor of $G^*$ in the ordered sense. We now note a second additional obstacle compared to the vertex-minor proof: the partitions of $G^*$ and $H$ do not necessarily agree. Our procedure thus follows two steps. Let $\hat V_1\subseteq L^*, \hat V_2\subseteq R^*, V^*_1\subseteq L^*,$ and $V^*_2\subseteq R^*$ be large disjoint subsets of $V(G)\setminus U$ such that $|\hat V_1| = |\hat V_2| = m$ and $|V^*_1| = |V^*_2| = m^*$ for $m,m^*$ chosen appropriately. Let $\hat V=\hat V_1\cup \hat V_2$ and $V^*=V^*_1\cup V^*_2$.

First, we align the partitions of $H$ and $G^*$. Recall that pivoting at a vertex will swap it between $L^*$ and $R^*$. The vertices in $U_1\cap R^*$ and $U_2\cap L^*$ are ``on the wrong side," so in \cref{lem:align_partition} we prove that with high probability, we can pivot on edges of $[V^*_1, U_1]$ and $[U_2,V^*_2]$ and delete $V^*$ to obtain $G=(L,R,E)$ such that $L\supseteq U_1$ and $R\supseteq U_2$. Furthermore, we guarantee that after these pivots, $G$ is still uniformly random, i.e.\ $G\sim \bG(L,R,1/2)$.

Next, we let $\hat V_1=\{v_1,\ldots, v_m\}\subseteq L$ and $\hat V_2=\{w_1,\ldots,w_m\}\subseteq R$. We perform a series of pivot operations inside $\hat V$ to ``mix" the induced graph on the set $U$. For each $J\subseteq [m]$, where $J=\{i_1<\ldots<i_s\}$, define:
\begin{equation}\label{eq:GJdefnbipartite}
G_J\coloneqq G\apiv v_{i_1}w_{i_1}\apiv v_{i_2}w_{i_2}\apiv\cdots \apiv v_{i_s}w_{i_s}.
\end{equation}
Note that the $\apiv$ operator is an involution and preserves bipartiteness, so $G_J[U]$ is still uniformly random. We also prove that if $J\symdif J'$ is large, then $G_J[U],G_{J'}[U]$ are sufficiently uncorrelated (see \cref{lem:JJ' are mixed bipartite}). Then using the second moment method, we have $H=G_J[U]$ for some $J$ with high probability. The final result follows from a union bound over all choices of $H$ and $U$.


We now describe how to align the partitions of $H$ and $G$. Our goal will be to swap every vertex in $Y_2=U_2\cap L^*$ and $Y_1=U_1\cap R^*$ to the other side of the bipartition such that the resulting graph is still uniformly random.

\begin{lemma}\label{lem:align_partition}
Let $G^*=(L^*,R^*,E^*)\sim \bG(\ell,r,1/2)$. Let $V_1^*,Y_2\subseteq L$, $V_2^*,Y_1\subseteq R$, be disjoint sets with $|V_1^*|,|V_2^*|=m$ and $|Y_1|,|Y_2|\leq k$. Then with probability at least $1-2m^k2^{-m}$ we may construct a pivot-minor $G=(L,R,E)$ of $G^*$ with $L=L^*\cup Y_1\setminus (V_1^*\cup Y_2)$ and $R=R^*\cup Y_2\setminus (V_2^*\cup Y_1)$ such that $G\sim \bG(L,R,1/2)$.
\end{lemma}
\begin{proof} We give an algorithm to construct $G$. Let $V_1^*=\{v_1,\ldots, v_m\}\subseteq L^*$ and $Y_1=\{y_1,\ldots, y_t\}\subseteq R^*$. We will reveal edges one at a time and construct a sequence of graphs $(G_i^*)_{i\geq 0}$ such that each $G_i^*$ is uniformly random on its partition. Let $G_0^*=G^*$.

First, we reveal whether $v_1y_1$ is in $E(G^*_0)$. If so, then we call this a ``success" and let $G^*_1=G^*_0\times v_1y_1-v_1$, so we have swapped $y_1\in L(G^*_1)$. Moreover, conditional on the edges we have revealed so far, by deleting $v_1$, $G^*_1$ consists of un-revealed edges and is uniformly random on its partition. If $v_1y_1\notin E(G^*_0)$, we call this a ``failure" and let $G^*_1=G^*_0-v_1$. By deleting $v_1$, note that $G^*_1$ is still uniformly random. Repeat this process with pairs $v_2y_1,v_3y_1,\ldots$ until a success occurs, then repeat with vertices $y_2,\ldots, y_t$ until $t$ successes occur. Finally, we obtain $G$ by deleting the remaining vertices of $V_1^*$. 

Observe that this algorithm fails precisely when among the $m$ vertices of $V_1^*$, we have fewer than $|Y_1|$ successes, where each trial succeeds independently with probability $1/2$. Thus the algorithm fails with probability $\PP[\rm{Bin}(m,1/2)<|Y_1|]<m^k2^{-m}$. 
Recall that conditional on the previous successful indices, each $G_i^*$ is uniformly random, and thus conditioned on the algorithm succeeding, the resulting graph is uniformly random. Repeat this algorithm on the sets $V_2^*$ and $Y_2$.
\end{proof}

Next, we adapt the mixing arguments of Section~\ref{section:randomwalks} to show that we can obtain $H$ inside the aligned graph $G$.

\begin{definition} Let $\cB_{\ell,r}$ be the set of ordered bipartite graphs on some fixed bipartition $(L,R)$ with sizes $\ell,r$. Suppose $\mu$ is a probability distribution on $\cB_{\ell,r}$ and let $\rm{bPiv}(\mu)$ be the the distribution of the graph $\Gamma\symdif K_{S,T}$, where $\Gamma\sim \mu$ and $S\subseteq L$ and $T\subseteq R$ are uniformly random and independent of $\Gamma$. Also, let $\mu_{\rm{unif}}$ be the uniform distribution on $\cB_{\ell,r}$
\end{definition}

Observe that the operator $\rm{bPiv}$ simulates the transformation of supposing there is an edge $vv'$ outside of $\Gamma$ with random neighborhoods in $\Gamma$ (respecting the partition), and then pivoting on $vv'$ and deleting $v,v'$.

\begin{lemma}\label{lem:eigenvalue_rank_bipartite}
The eigenvectors of $\rm{bPiv}$ are given by $\{\chi_G \mid G \in \cB_{\ell,r}\}$ where $\chi_G \in \mathbb{R}^{\cB_{\ell,r}}$ is defined as $\chi_G(H) =(-1)^{|E(G) \cap E(H)|}$. For each $G \in \cB_{\ell,r}$, let $\lambda_G$ denote the eigenvalue of $\chi_G$ under $\rm{bPiv}$, and let $A\in \bF_2^{r\times \ell}$ be the bipartite adjacency matrix of $G$. Then we have
\[\lambda_G = 2^{-\rm{rank}_{\bF_2}(A)}.\]
\end{lemma}
\begin{proof}
Similarly to \cref{lem:eigenvalues_of_walk_explicit}, we have
\[
\rm{bPiv}(\chi_G)(H)=2^{-\ell-r}\sum_{S,T}(-1)^{|E(H\symdif K_{S,T})\cap E(G)|}= \chi_G(H)\cdot 
2^{-\ell-r}\sum_{S,T}(-1)^{|E(G)\cap E(K_{S,T})|}
\]
Thus $\chi_G$ is an eigenvector with eigenvalue $\lambda_G = \bE_{S,T}[(-1)^{|E(G) \cap E(K_{S,T})|}]$. Then, analogous to \cref{lem:eigenvalue_rank}, if we let $x_S,x_T$ be the indicators for $S$ and $T$, we have 
\begin{equation*}
\lambda_G=2^{-\ell-r}\sum_{\substack{S\subseteq L\\T\subseteq R}}(-1)^{x_T^\top Ax_S^{}}=2^{-\ell-r}\sum_{x_S\in \bF_2^{\ell}}\bigg(\sum_{x_T\in\bF_2^r}(-1)^{x_T^\top Ax_S^{}}\bigg).
\end{equation*}
If $x_S\in \rm{Null}(A)$ then the inner sum is $2^r$, and otherwise it is $0$. Thus $\lambda_G=2^{-\ell-r}\cdot 2^r|\rm{Null}(A)|=2^{-\rm{rank}(A)}$.
\end{proof}

We now conclude that the small eigenvalues cause this random walk to converge quickly to the uniform distribution.

\begin{lemma}\label{lem:bip_random_walk_mixing}
Let $\mu_0 = \delta_{I_{\ell,r}}$ be the point mass distribution on the independent graph $I_{\ell,r} \in \cB_{\ell,r}$. Then for $t=1,2,\ldots$ define  $\mu_t = \rm{bPiv}(\mu_{t-1})$. For any graph $H \in \cB_{\ell,r}$ we have the inequality
$$|\mu_\ell(H) - \mu_{\rm{unif}}(H)| \leq 2^{- t}.$$
\end{lemma}
The proof is similar to \cref{lem:random_walk_mixing}, using the simpler observation that every eigenvector other than the steady state has eigenvalue at most $1/2$.

Next, we argue that this mixing process is exactly simulated by pivoting on a sequence of edges outside of the set. We require a subtle independence lemma analogous to \cref{lem:complementing_independence} which shows that pivoting on an edge in $\bG(\ell,r,1/2)$ does not affect the distribution of the rest of the graph, even if we condition on other information. Recall that for sets $X,Y$ we define $[X,Y]=\{\{x,y\}\mid x\in X, y\in Y\}$.

\begin{lemma}\label{lem:pivot_independence_bipartite}
Let $L,R$ be disjoint sets; let $W_1\subseteq \hat V_1\subseteq L$ and $U_1=L\setminus \hat V_1$, and let $W_2\subseteq \hat V_2\subseteq R$ and $U_2=R\setminus \hat V_2$. Fix an ordered bipartite graph $\hat G=(\hat V_1,\hat V_2,\hat E)$ and let $e_1,\ldots, e_t$ be a (possibly repeating) sequence of pairs contained in $W_1 \cup W_2$ such that $\hat G\times (e_i)_{i=1}^t$ is defined. Partition $[L,R]$ into 
\[E_1=[\hat V_1,\hat V_2]\cup [W_1,U_2]\cup [U_1,W_2]\text{\quad and \quad}E_2=[U_1,U_2]\cup [\hat V_1\setminus W_1,U_2]\cup [U_1,\hat V_2\setminus W_2].\]
Let $G$ be a random graph drawn from $\bG(L,R,1/2)$ conditioned on $G[\hat V]=\hat G$ and let $G'=G\times (e_i)_{i=1}^t$. Then $E(G')\cap E_2$ is uniformly random and independent of $E(G)\cap E_1$.
\end{lemma}

Note that as defined above, $L(G')\supseteq L\setminus W_1$ and $R(G')\supseteq R\setminus W_2$, so $E_2$ respects the partition of $G'$ and hence it makes sense to say that $E(G)\cap E_2$ is uniformly random.
\begin{proof}
We sketch the same argument as \cref{lem:complementing_independence}. Fix a subset $A\subseteq E_1$ and condition on $E(G)\cap E_1=A$; we prove by induction on $t$ that $E(G')\cap E_2$ is uniformly random. The base case $t=0$ follows by definition of $\bG(L,R,1/2)$.

Suppose $t\geq 1$ and let $G''=G\times(e_i)_{i=1}^{t-1}$ where by induction $E(G'')\cap E_2$ is uniformly random. Then letting $e_t=\{v,w\}$, $G'=G''\times e_t$ is obtained by the involutive map $H\mapsto H\triangle K_{S,T}$ where $S=N_{G''}(v)$ and $T=N_{G''}(w)$ are determined by $A=E(G)\cap E_1$. This induces an involution $E(G'')\cap E_2\mapsto E(G')\cap E_2$ and therefore preserves the uniform measure.
\end{proof}

\begin{lemma}\label{lem:complementing_mixing_bipartite}
Let $G\sim \bG(\ell,r,1/2)$. Let $U_1, \hat V_1\subseteq L(G)$ and $U_2,\hat V_2\subseteq R(G)$ be disjoint where $|U_1|=k_1$ and $|U_2|=k_2$. Let $\hat G=(\hat V_1,\hat V_2,\hat E)$ be an ordered bipartite graph and let $(v_iw_i)_{i=1}^D$ be a sequence of disjoint pairs in $[\hat V_1,\hat V_2]$ such that $\hat G\times (v_iw_i)_{i=1}^D$ is well-defined. Then
\[
\Big|\PP\Big[G\times (x_i)_{i=1}^D[U_1,U_2]=G[U_1,U_2]\ \Big|\ G[\hat V]=\hat G\Big]-2^{-k_1k_2}\Big|\leq 2^{-D}
\]
\end{lemma}
\begin{proof}
We follow \cref{lem:complementing_mixing} with adjusted notation. Let $U=U_1\cup U_2$ and $\hat V=\hat V_1\cup \hat V_2$. For times $t=0,1,\ldots, D$, define $G_t\coloneqq G\times (v_iw_i)_{i=1}^t$ and define $\mu_t\in [0,1]^{\cG_k}$ as the distribution of $H_t:=(G_t\symdif G)[U]$, conditioned on $G[\hat V]=\hat G$. Then we have the point mass $\mu_0=\delta_{I_{k_1,k_2}}$.

At time $t$, let $W_1=\{v_i\mid i<t\}$ and $W_2=\{w_i\mid i<t\}$. Observe that $H_{t-1}$ is a function of $G[\hat V]=\hat G$ and $N_G(u)$ for $u\in W_1\cup W_2$. Then we have $H_t=H_{t-1}\symdif K_{S,T}$ where $S=N_{G_{t-1}}(v_t)\cap U_2$ and $T=N_{G_{t-1}}(w_t)\cap U_1$. By \cref{lem:pivot_independence_bipartite} (applied with $W_1,W_2,\hat V_1,\hat V_2,U_1,U_2$), we have that $S$ (resp. $T$) is a uniformly random subset of $U_1$ (resp. $U_2$), and conditionally independent of $G[\hat V]$ and $N_G(u)$ for $u\in W_1\cup W_2$. Thus $S$ (resp. $T$) is independent of $H_{t-1}$. Then by definition, we have $\mu_t=\rm{bPiv}(\mu_{t-1})$ for all $t$, so by \Cref{lem:bip_random_walk_mixing},
\begin{align*}
    \Big|\PP\Big[G\circ (x_i)_{i=1}^D[U]=G[U]\ \Big|\ G[\hat V]=\hat G\Big]-2^{-k_1k_2}\Big| &= |\mu_D(I_{k_1,k_2}) - \mu_{\rm{unif}}(I_{k_1,k_2})|\leq 2^{- D}
\end{align*}
as desired.
\end{proof}

We would now like to prove the analogue of \cref{lem:JJ' are mixed} to relate two mixed graphs $G_J,G_{J'}$, but we first need an analogue of the reordering \cref{lem:minor-reordering}. The following analogue of \cref{thm:structure of minors} follows from \cref{cor:pm_equiv_to_matroid_minor}.

\begin{corollary}
    Let $H,G$ be a ordered bipartite graphs and let $H$ be a pivot-minor of $G$. Then for every $v \in V(G) \setminus V(H)$, $H$ is a pivot-minor of
    \begin{enumerate}
        \item $G - v$, or
        \item $G \times vu - v$ for some $u \in N(v)$.
    \end{enumerate}
\end{corollary}
\begin{proof}
This follows from the analogous fact about matroids: If $N$ is a minor of $M$ and $e$ is an element in the ground set of $M$ but not $N$, then $N$ is a minor of $M/e$ or $M \setminus e$.
\end{proof}

From this we get the following ``reordering lemma" analogue of \cref{lem:minor-reordering}. 

\begin{lemma}\label{lem:bipartite-pivot-reordering}
    Fix a set $\hat V$. Let $\hat{G}$ be an ordered bipartite graph with $V(\hat G) = \hat V$. Let $v_1w_1, \dots, v_sw_s$
    be a (possibly repeating) sequence of pairs such that $\hat G \times v_1w_1 \times \cdots \times v_sw_s$ is defined. Then there is a sequence $e_1,\dots, e_D$ of pairs in $\hat V$ such that for any ordered bipartite graph $G$ with $\hat V \subseteq V(G)$ and $G[\hat V] = \hat{G}$, we have
    \begin{enumerate}[(I)]
    \item[\emph{(I)}] $G\times v_1w_1 \times \cdots \times v_sw_s-\hat V =G\times e_1\times \cdots \times e_D-\hat V$.
    \item[\emph{(II)}] The pairs $e_1,\ldots, e_D$ are pairwise disjoint.
    \item[\emph{(III)}] If $v\in \hat V$ appears exactly once in $(v_i)_{i=1}^s \cup (w_i)_{i=1}^s$, then $v\in e_j$ for some $j\in [D]$.
    \end{enumerate}
\end{lemma}

The proof of \cref{lem:bipartite-pivot-reordering} is similar to the proof of \cref{lem:minor-reordering}. In particular, it uses the following inductive lemma, analogous to \cref{lem:inductive_reordering}.
\begin{lemma}
For any ordered bipartite graph $G$ with $V(G) = \hat V \disjcup U$, sequence of pairs $v_1w_1, \dots, v_s w_s$ such that $G \times v_1w_1 \times \dots \times v_s w_s$ is defined, and special vertex $v \in \hat V$, there exists a sequence $(y_i)_{i=1}^D$ satisfying (I) and (II) above such that either $v\notin \bigcup_iy_i$ or $v\in y_1$.
\end{lemma}

Applying the above to the appropriate ``gadget" $\gadget$ as in the proof of \cref{lem:minor-reordering} yields the desired result. We now apply the $\rm{bPiv}$ random walk to $G_J$,$G_{J'}$ to obtain the analogue of \cref{lem:JJ' are mixed}.

\begin{lemma}\label{lem:JJ' are mixed bipartite} Let $G=(L,R,E)\sim \bG(\ell,r,1/2)$ and $U\subseteq V(G)$ such that $|U\cap L|=k_1$ and $|U\cap R|=k_2$. Let $\hat V_1=\{v_1,\ldots, v_m\}\subseteq L\setminus U$ and $\hat V_2=\{w_1,\ldots, w_m\}\subseteq R\setminus U$. Let $J,J'\subseteq [m]$ and define $G_J,G_{J'}$ as in Equation \eqref{eq:GJdefnbipartite}. Then we have
\[
\big|\PP[G_J[U]=G_{J'}[U]]-2^{-k_1k_2}\big|\leq (3/4)^{-|J\symdif J'|}
\]
\end{lemma}

\begin{proof}
Write $J=\{i_s<\ldots< i_1\}$ and $J'=\{i_{s+1}<\ldots<i_{s+s'}\}$. Fix some $\hat G\in \cB_{m,m}$ and condition on $G[\hat V_1,\hat V_2]=\hat G$. Then we have
\begin{align}
G_{J'}&=G_J\apiv (v_{i_t}w_{i_t})_{t=1}^s\apiv (v_{i_t}w_{i_t})_{t=s+1}^{s+s'}\notag\\
&=G_J\apiv (v_{i_t}w_{i_t})_{t=1}^{s+s'}\label{eq:pmGJ'fromGJ}\\
&=G_J \times (v_{j_t}w_{j_t})_{t=1}^{s''},\label{eq:pmGJ'fromGJ filtered_bipartite}
\end{align}
where $(j_t)_t$ is the subsequence of successful pivot attempts, which is a function $\hat G$. Let $X= X(\hat G)$ be the number of indices $j$ which occur exactly once in $(j_t)_t$ and such that $j\in J\symdif J'$. By \cref{lem:bipartite-pivot-reordering}, there exists $D = D(\hat G)$ and a sequence $(e_i)_{i=1}^{D}$ with properties (I)-(III), in particular  
$G_{J'}[U]=G_J\times(e_i)_{i=1}^{D}[U]$. By property (III), we have $D\geq X$, since $v_{j_t},w_{j_t}\in \bigcup_ie_i$ for every index $j_t$ appearing exactly once. By \cref{lem:complementing_mixing_bipartite},
\[\Big|\PP\big[G_J[U]=G_{J'}[U]\ \big|\  G[\hat V]=\hat G\big]-2^{-k_1k_2}\Big|\leq 2^{-D(\hat G)}\leq 2^{-X(\hat G)}.
\]
We now consider the unconditioned $G[\hat V]\sim \bG(m,m,1/2)$, and we prove that $X(G[\hat V])\sim \bin(|J\symdif J'|,1/2)$. Consider the sequence of attempted pivots in Equation \eqref{eq:pmGJ'fromGJ}. Now, for any index $j\in J\symdif J'$, there is some unique $t^*$ such that $i_{t^*}=j$. Moreover, in the graph $G'=G_J\apiv (v_{i_t}w_{i_t})_{t=1}^{t^*-1}$, we have the edge $v_jv_j'\in E(G')$ with probability 1/2. Therefore $j$ appears in the subsequence $(j_t)_t$ of successful pivot attempts with probablity $1/2$. By \cref{lem:pivot_independence_bipartite}, these edge events are mutually independent, so $X(G[\hat V])\sim \bin(|J\symdif J'|,1/2)$. Finally, using the law of total probability over $G[\hat V]$,
\begin{align*}
\Big|\PP\big[G_J[U]=G_{J'}[U]\big]-2^{-k_1k_2}\Big|
&=\Big|\Ex_{G[\hat V]}\big[\PP[G_J[U]=G_{J'}[U]\mid  G[\hat V]]-2^{-k_1k_2}\big]\Big|\\
&\leq \Ex_{G[\hat V]}\Big[\big|\PP[G_J[U]=G_{J'}[U]\mid  G[\hat V]]-2^{-k_1k_2}\big|\Big]\\
&\leq \Ex_{G[\hat V]}\big[2^{-X(G[\hat V])}\big]\\
&=(3/4)^{|J\symdif J'|},
\end{align*}
where the last equality uses the moment generating function of the binomial distribution.
\end{proof}

Next we provide the following fact about independence, similar to \cref{prop:symdiffindep}.

\begin{proposition}\label{prop:symdiffindep_bipartite}
Let $G\sim \bG(\ell,r,1/2)$, $U\subseteq V(G)$, and let $\hat V_1\subseteq L\setminus U$ and $\hat V_2\subseteq R\setminus U$. Let $J,J'\subseteq [m]$ and define $G_J,G_{J'}$ as in Equation \eqref{eq:GJdefnbipartite}. Then the graphs $G_J[U]$ and $(G_J\symdif G_{J'})[U]$ are independent.
\end{proposition}
The proof is identical to \cref{prop:symdiffindep}, by sampling $G_J$ first and then obtain $G_{J'}$ via Equation~\eqref{eq:pmGJ'fromGJ}. Next, we apply the second moment method to show that each ordered bipartite graph is a labeled pivot-minor of a random graph with high probability.

\begin{lemma}\label{lem:pm_high_probability} Let $G^*\sim \bG(\ell,r,1/2)$. Let $U\subseteq V(G)$ and $H=(U_1,U_2,E_H)$ be an ordered bipartite graph on $U$ and $m,m^*\geq 0$ such that $\ell,r\geq m+m^*+|U|$. Then $H$ is a labeled pivot-minor of $G$ (in the ordered sense) with probability at least $1-2^{|U_1||U_2|}(7/8)^m-2^{-m^*+k\log_2(m^*)+1}$.
\end{lemma}
\begin{proof} Let $G^*=(L^*,R^*,E^*)$ and let $V_1^*\subseteq L^*\setminus U$ and $V_2^*\subseteq R^*\setminus U$ have size $m^*$. By \cref{lem:align_partition}, with probability $1-2^{-m^*+k\log_2(m^*)+1}$ we can obtain a pivot-minor $G=(L,R,E)$ such that $U_1\subseteq L$ and $U_2\subseteq R'$ and $G\sim \bG(L,R,1/2)$.

Suppose that the aligning step succeeds. Now, let $\hat V_1\subseteq L\setminus U_1$ and $\hat V_2\subseteq R\setminus U_2$ have size $m$. For $J\subseteq [m]$, define $G_J$ as in Equation~\eqref{eq:GJdefnbipartite} and let $X=\sum_{J\subseteq [n]}X_J$, where $X_J=\bbone_{G_J[U]=H}$. Note that each $G_J[U]$ is uniformly random and thus we have probability of success $\PP[G_J[U]=H]=2^{-|U_1||U_2|}$. By linearity of expectation, $\Ex[X]=2^{-|U_1||U_2|}2^m$. We also have
\begin{align*}
\var[X]&=\sum_{J,J'\subseteq [m]}\Ex[X_JX_{J'}]-\Ex[X_J]\Ex[X_{J'}]\\
&=2^{-|U_1||U_2|}\sum_{J,J'\subseteq [m]}\PP[G_J[U]=G_{J'}[U]]-2^{|U_1||U_2|},
\end{align*}
using \cref{prop:symdiffindep_bipartite}. Then by \cref{lem:JJ' are mixed bipartite} and the binomial theorem,
\[
\var[X]\leq 2^{-|U_1||U_2|}\sum_{J\subseteq [m]}\sum_{i=0}^m{m\choose i}(3/4)^{-i}=2^{-|U_1||U_2|}2^m(7/4)^m.\]
Finally, by Chebyshev's inequality (\cref{lem:chebyshev}),
\[
\PP[X=0]\leq \frac{\var[X]}{\Ex[X]^2}=2^{-|U_1||U_2|+m}(7/4)^m\cdot 2^{2|U_1||U_2|-2m}=2^{|U_1||U_2|}(7/8)^n.
\]
Thus with probability at least $1-2^{-m^*+k\log_2(m^*)+1}-2^{|U_1||U_2|}(7/8)^n$, both steps succeed, implying that $H$ is a pivot-minor of $G$.
\end{proof}

\begin{proof}[Proof of \cref{thm:bip_pm_universal}]
Let $G\sim\bG(\ell,r,1/2)$, where $\ell\geq r= (1+c)Ck^2\geq 1$, where $C=\frac{1}{2\log_2(8/7)}+\frac 14$ and $c=c(k)$ is bounded below by some positive constant. Let $m=(1+c)\frac{k^2}{2\log_2(8/7)}$ and $m^*=(1+c)k^2/4-k$, so then $\ell,r\geq m+m^*+k$. By \cref{lem:pm_high_probability}, the probability that a given ordered bipartite graph $(H,U_1,U_2)$ on $k$ vertices of $V(G)$ is not a vertex-minor of $G$ is at most $2^{|U_1||U_2|}(7/8)^m+2^{-m^*+k\log_2(m^*)+1}$. Observe that
\[2^{|U_1||U_2|}(7/8)^m\leq 2^{k^2/4}2^{-(1+c)k^2/2}=2^{-(1+2c)k^2/4} \leq 2^{-(1+c)k^2/4},
\]
and $2^{-m^*+O(k\log_2(m^*))}=2^{-(1+c+o(1))k^2/4}$. By a union bound over all ${\ell+r\choose k}\sum_{i=0}^k{k\choose i}2^{i(k-i)}<(\ell+r)^k2^{k+k^2/4}$ choices of $H$,
\begin{align}
\PP[G\text{ is not $k$-p.m.\ universal}]&\leq (\ell+r)^k2^{k+k^2/4}\cdot 2^{-(1+c+o(1))k^2/4}\notag\\
&= \exp_2\left(k\log_2(\ell+r)-(c+o(1))k^2/4\right).\notag\label{eq:pm-final-bound-bipartite}
\end{align}
If $\ell\leq 2^{o(\sqrt r)}$, then $k\log(\ell+r)=ko(\sqrt{ck^2})=o(ck^2)$ and thus $G$ is $k$-pivot-minor universal with probability at least $1-2^{-(c+o(1))k^2/4}$. 
\end{proof}
\begin{remark} If we weaken the definition of $k$-universality to require that we contain every graph (resp. matroid) of size $k$ \textit{up to isomorphism}, then the analogous Theorems to \ref{thm:bip_pm_universal} and \ref{thm:stronger_matroid_universal} hold without the condition $\max(\ell,r)\leq 2^{o(\sqrt{\min(\ell,r)})}$ (by restricting to a balanced subgraph with $\ell=r$). Moreover, the constant could be improved to $C=\frac{1}{2\log_2(8/7)}\approx 2.60$ by skipping the ``aligning step" (\cref{lem:align_partition}).

Furthermore, if we are asking whether graph $G$ (resp.\ matroid $M$) contains just one particular minor (labeled or up to isomorphism) with high probability, the union bound above can be reduced to improve the constant to $C=\frac{1}{4\log_2(8/7)}\approx 1.30$. In general, the constants in this section could likely be improved with more work.
\end{remark}

\section{Vertex-minor Ramsey and concluding remarks}\label{section:VM Ramsey conclusion}

We now prove our bounds on the vertex-minor Ramsey problem stated in the introduction.

\ramseybounds*

\begin{proof} We first prove the upper bound. When $k=0,1$ this is satisfied with equality. If $k\geq 2$, let $N=R_{\rm{vm}}(k)-1$, so there is a graph $G$ on $N$ vertices with $U \subseteq V(G)$ an independent set in $G$ such that $|U| \leq k-1$ and $G$ does not contain a larger independent set as a vertex-minor. For all $S \subseteq U$, let $V_S = \{v \in V(G) \setminus U : N(v) \cap U = S\}$. Then we have the disjoint union $V(G) = U \cup \bigcup_{S \subseteq U} V_S$. We now claim that $|V_S| \leq \min\{2,|S|\}$ for all $S \subseteq U$. Assuming this, we conclude
$$R_{\rm{vm}}(k) - 1 = N \leq (k-1) + 1 \cdot \binom{k-1}{1} + 2 \cdot (2^{k-1}-k) = 2^k-2.$$

Note that if $x \in V_{\emptyset}$, then $U\cup \{x\}$ is an independent set, contradicting the maximality of $U$. Suppose towards a contradiction that $x,y \in V_{\{u\}}$ for some $u \in U$. By locally complementing at $u$, we may assume that $xy \not\in E(G)$. Then $U \cup \{x,y\} \setminus \{u\}$ is an independent set, contradicting the maximality of $U$.

Suppose towards a contradiction that $|S| \geq 2$ and $x,y,z \in V_S$. Let $u \in S$. By locally complementing at $u$, we may assume that $G[\{x,y,z\}]$ contains at most 1 edge. If $\{x,y,z\}$ is an independent set in $G$, then $U \cup \{y,z\} \setminus \{u\}$ is an independent set in $G \times ux$. Otherwise without loss of generality $yz \in E(G)$. Then $U \cup \{y,z\} \setminus \{u\}$ is an independent set in $G * x * u$. In every case we obtain an independent set of size larger than $U$ as a vertex-minor, a contradiction.

For the lower bound, let $G\sim \bG(n,1/2)$. Note that $G$ contains $I_k$ as a vertex-minor if and only if $G$ is locally equivalent (l.e.) to some graph $H$ with $\alpha(H)\geq k$. A result of Bahramgiri and Beigi \cite{Bahramgiri2007EnumeratingTC} states that for any graph $H$ on $n$ vertices, there are at most $3^n$ graphs locally equivalent to $H$. Then, by a union bound,
\begin{align*}
\PP[G\text{ l.e.\ }H,\ \alpha(H)\geq k]
&\leq 3^n\PP[\alpha(G)\geq k]\\
&\leq 3^n{n\choose k}2^{-{k\choose 2}}
\leq \exp_2\big(n\log_23+k\log_2n-k^2/2+k/2\big)
<1,
\end{align*}
as long as $n\leq (1+o(1))\frac{1}{2\log_23}k^2$.
\end{proof}
We note that our upper bound is the correct answer for $k \leq 3$. (For $k=3$, it is not hard to verify that the wheel on 6 vertices does not contain an independent of size 3 as a vertex-minor.)
However, 
since the lower bound in \cref{thm:ramsey_bounds} comes from the random graph, it is natural to believe that the polynomial lower bound is closer to the truth for sufficiently large $k$ than our modest upper bound. This belief leads us to \cref{conj:ramsey_conj}, restated below.

\ramseyconj*

By \cref{thm:vm_universal}, the random graph $\bG(n,1/2)$ cannot be a counterexample, and indeed it is possible that $R_{\rm{vm}}(k) = O(k^2)$. We note that proper vertex-minor closed classes have the Erd\H{o}s-Hajnal property~\cite{chudnovsky2018vertex}. This implies that if we restrict our graphs to any proper vertex-minor closed class, they contain an independent set of size $\Omega(n^c)$ as a vertex-minor for some $c>0$ possibly depending on the class. We also make the connection between pivot-minor universality and its corresponding Ramsey-type question.

\begin{definition}
Let $R_{\rm{piv}}(k)$ be the minimum $n$ such that any $n$-vertex graph contains an independent set or clique of size $k$ as a pivot-minor.
\end{definition}

Note that unlike the definition of $R_{\mathrm{vm}}(k)$, we need to allow either a clique or an independent set in this definition since $K_n$ does not contain any non-trivial independent set as a pivot-minor. 

We note that trivially $R_{\rm{vm}}(k-1) \leq R_{\rm{piv}}(k) \leq R(k)$, the diagonal Ramsey number. \cref{thm:ER_pm_universal} supports the following conjecture, which would imply \cref{conj:ramsey_conj}.

\begin{conjecture}
    $R_{\rm{piv}}(k) \leq \mathsf{poly}(k)$.
\end{conjecture}

We note that proper pivot-minor closed classes also have the Erd\H{o}s-Hajnal property \cite{davies2025pivot}, so this conjecture holds if we restrict to any proper pivot-minor closed class. We also improve the trivial upper bound of $R(k)$ for $R_{\rm{piv}}(k)$.

\begin{theorem}
    For all $k \geq 0$, $R_{\rm{piv}}(k) \leq (k-1)(2^k - 1) +1$.
\end{theorem}
\begin{proof}
    Let $G$ be a graph with no independent set or clique of size $k$ as a pivot-minor on $R_{\rm{piv}}(k)-1$ vertices. By possibly replacing $G$ with a pivot-equivalent graph, we let $U$ be an independent set in $G$ such that $G$ does not contain a larger independent set as a pivot-minor. For all $S \subseteq U$, let $V_S = \{v \in V(G) \setminus U : N(v) \cap U = S\}$. Clearly $V_{\emptyset}$ is empty by the maximality of $U$. For $S$ containing a vertex $u$, we claim that $V_S$ can not contain an independent set of size 3. Indeed otherwise if $\{x,y,z\} \in V_S$ is independent, then $U \cup \{y,x\} \setminus \{u\}$ is an independent set in $G \times ux$. We also note that $V_S$ can not contain an induced path on 3 vertices. Indeed if $x,y,z$ are the vertices of an induced path in $V_S$ in that order, then $U \cup \{y,x\} \setminus \{u\}$ is an independent set in $G \times ux$. Clearly $V_S$ can not contain a clique of size $k$ by definition of $G$. It follows that $G[V_S]$ is a disjoint union of cliques and thus $|V_S| \leq 2(k-1)$. Thus
    \begin{equation*}R_{\rm{piv}}(k) - 1 \leq (k-1) + (2^{k-1} - 1)(2k-2) = (k-1)(2^k - 1). \qedhere\end{equation*}
\end{proof}

We make some final conjectures about vertex-minor universality. \cref{thm:vm_universal} and our proof of the lower bound in \cref{thm:ramsey_bounds} together show that for any $\eps>0$, the random graph $G\sim \bG(n,1/2)$ is with high probability $(\alpha-\eps)\sqrt{n}$-vertex-minor universal but not $(\beta+\eps)\sqrt{n}$-vertex-minor universal, where $\alpha = \sqrt{2\log_2(4/3)} \approx 0.911$ and $\beta = \sqrt{2\log_2(3)n} \approx 1.78$. It is natural to conjecture that there is a sharp threshold for universality, as follows.
\begin{conjecture}
There exists $C\in [\sqrt{2\log_2(4/3)},\sqrt{2\log_2(3)}]$ such that for any $\eps > 0$, with high probability $G\sim \bG(n,1/2)$ is $(C-\eps)\sqrt n$-vertex-minor universal and not $(C+\eps)\sqrt n$-vertex-minor universal.
\end{conjecture}

It is also natural to ask for what values of $p$ the random graph $\bG(n,p)$ is $k$-vertex-minor universal for $n$ polynomial in $k$. The main difficulty in extending our approach to $p \neq 1/2$ is that the random walk defined in \cref{section:randomwalks} now has steps which are correlated to each other: after flipping an edge which was not yet revealed, its probability of existing is now $1-p$ instead of $p$. If one could get around this issue, the resulting random walk would then have the second largest eigenvalue bounded in absolute value by the max of $1-2p^2, 1-2(1-p)^2,$ and $1 - 2p(1-p)$. Assuming without loss of generality that $p \leq 1/2$, this means it would suffice to run the random walk for $O(k^2/p^2)$ steps. Thus we conjecture the following. 

\begin{conjecture}\label{conj:p neq 1/2}
    If $p = \omega(1/\sqrt{n})$ is at most $1/2$ and $G \sim \bG(n,p)$ or $G \sim \bG(n,1-p)$, then $G$ is $k$-vertex-minor universal with high probability for some $k = \Omega(p\sqrt n)$.
\end{conjecture}

If true, this would qualitatively generalize \cref{thm:vm_universal}, although unlike \cref{thm:vm_universal} it is not necessarily tight up to a constant for all $p$. \cref{conj:p neq 1/2} would also imply that $\bG(n,p)$ is not a counterexample to \cref{conj:ramsey_conj} for any $p$. This follows for instance by noting that if $p$ is outside the range $[n^{-1/4}, 1-n^{-1/4}]$, then $\bG(n,p)$ already contains a clique or independent set of size $\Omega(n^{1/4}\log n)$.

\vspace{-0.2cm}

\paragraph{Acknowledgments.} 
This work began at the Fall 2025 Atlanta Graduate Combinatorics Workshop at Georgia Tech. We thank Georgia Tech for providing funding for the workshop, Sean Longbrake for helping to organize it, and Elliot Frederickson for his motivating presence there. We also extend our thanks to Xiaoyu He and Rose McCarty for enlightening conversations, and to Rose McCarty for telling us about the problem of whether the uniformly random graph is vertex-minor universal (which was originally posed by Nathan Claudet) and pointing us to relevant sources. Finally, we thank Nathan Claudet for helpful comments on an earlier version of this paper.

\end{document}